\documentclass[sigconf,nonacm]{acmart}

\usepackage{hyperref}   
\usepackage{tikz}
\usepackage{amsmath}
\usepackage{amsfonts} 
\usepackage{amsthm}
   
\usepackage{forloop}
\usepackage{adjustbox}
\usepackage[graphicx]{realboxes}
\usepackage{rotating}
\usepackage{mdframed}
\usepackage[most]{tcolorbox}
\usepackage{todonotes}    
\usepackage{url}

\usepackage{mathtools,halloweenmath}

\def\C{
\mathsf C}

\makeatletter
\newcommand{\add@empty@sup}{\@ifnextchar^{}{^{}}}

\usepackage{soul}
\usepackage{enumerate}
\usepackage{subfigure}
\usepackage{graphics,graphicx,epsfig}
\usepackage{caption}
\usepackage{subcaption}
\usepackage{threeparttable}
\usepackage{stfloats}
\usepackage{multirow}

\usepackage{bm}
\usepackage{color}

\usepackage{multirow}
\usepackage{subeqnarray}
\usepackage{cases}
\usepackage{colortbl}
\usepackage{wasysym}
\usepackage[T1]{fontenc}    
\usepackage{bbm}
\usepackage{graphicx}
\usepackage{cleveref}

\usepackage{url}            
\usepackage{booktabs}       
\usepackage{amsfonts}       
\usepackage{nicefrac}       

\usepackage{microtype}
\usepackage{algorithmic} 
\usepackage{makecell}
\usepackage[ruled,linesnumbered,vlined]{algorithm2e}

\usepackage{enumitem}
\newlist{Properties}{enumerate}{2}
\setlist[Properties]{label=Property \arabic*.,itemindent=*}

\usepackage{mathtools}
\DeclarePairedDelimiter{\ceil}{\lceil}{\rceil}
\usepackage{adjustbox}
\usepackage{cuted}
\usepackage{xcolor}
\usepackage{enumitem,kantlipsum}
\usepackage{diagbox}
\usepackage{array}
\usepackage{cellspace} 
\definecolor{myred}{RGB}{234, 107, 102}
\definecolor{mypur}{RGB}{166, 128, 184}

\definecolor{Blue1}{RGB}{240,248,255} 
\definecolor{Blue2}{RGB}{70,130,180}

\usepackage[n,advantage, operators, sets, adversary, landau, probability, notions, ff, mm, primitives, events, complexity, oracles, asymptotics, keys]{cryptocode}

\newcommand{\eg}{{\it e.g.}\xspace}
\newcommand{\etal}{{\it et~al.}\xspace}

\newcommand{\ie}{{\it i.e.}\xspace}

\newcommand{\SIM}{\mathsf{Sim}}

\newtheorem{property}{Property}
\newtheorem{theorem}{Theorem}
\newtheorem{lemma}[theorem]{Lemma}

\newtheorem{definition}{Definition}
\newcommand{\figbox}[3]
{
\begin{figure}[t]

\fbox{
\centering
\begin{minipage}{0.46\textwidth}
  \begin{center}
  \underline{\centering #1 }
\end{center}

  {#2}

\end{minipage}
}

\caption{#3}
\end{figure}
}

\newcommand{\Sim}{\mathtt{Sim}}
\newcommand{\Adv}{\mathcal{A}}

\newcommand{\PA}{\ensuremath{{\sf P}_{A}}}
\newcommand{\PB}{\ensuremath{{\sf P}_{B}}}
\newcommand{\ID}{{\sf ID}}

\newcommand{\AGen}{\mathcal{HE}.{\sf Gen}}
\newcommand{\AEnc}{\mathcal{HE}.{\sf Enc}_{\pk}}
\newcommand{\ADec}{\mathcal{HE}.{\sf Dec}_{\sk}}
\newcommand{\ASum}{\mathcal{HE}.{\sf Sum}}
\newcommand{\AFresh}{\mathcal{HE}.{\sf Refresh}}
\newcommand{\VIEW}{{\sf View}}
\newcommand{\keyA}[1]{k^{#1}_A}
\newcommand{\keyB}[1]{k^{#1}_B}
\newcommand{\id}{{\sf id}}

\newcommand{\Mone}{\mathcal{M}_{\tt one}^{\tau}}
\newcommand{\Mwf}{\mathcal{M}_{\tt mult}}
\newcommand{\F}[1]{\mathcal{F}_{\textit{#1}}}
\newcommand{\Prot}[1]{\Pi_{\textit{#1}}}
\newcommand{\usample}{\xleftarrow{\tt{\$}}}
\newcommand{\sid}{{\sf sid}}
\newcommand{\PRF}{{durPRF}}
\newcommand{\cPSI}{{CPSI}}

\newcommand\algorithmicprocedure{\textbf{procedure}}
\newcommand{\algorithmicendprocedure}{\algorithmicend\ \algorithmicprocedure}
\makeatletter
\newcommand\PROCEDURE[3][default]{%
  \ALC@it
  \algorithmicprocedure\ \textsc{#2}(#3)%
  \ALC@com{#1}%
  \begin{ALC@prc}%
}
\newcommand\ENDPROCEDURE{%
  \end{ALC@prc}%
  \ifthenelse{\boolean{ALC@noend}}{}{%
    \ALC@it\algorithmicendprocedure
  }%
}
\newenvironment{ALC@prc}{\begin{ALC@g}}{\end{ALC@g}}
\makeatother

\usepackage{cryptocode}
\usepackage{dashbox}

\createprocedureblock{interactGame}{center,boxed}{}{}{}

\NewDocumentEnvironment{gameInteract}{m+b}
{%
  \par\begin{Spacing}{1.2}
    \interactGame[linenumbering]{#1}{#2}
    \par\addvspace{\baselineskip}
  \end{Spacing}
}{}

\newcommand{\Comb}[2]{{}^{#1}\C_{#2}}

\usepackage{tikz}
\usepackage{array}
\usepackage{colortbl}

\begin{document}

\title{\bf PrivacyGo: Privacy-Preserving Ad Measurement with Multidimensional Intersection}

\author{{Jian Du, Haohao Qian, Shikun Zhang, Wen-jie Lu,} Donghang Lu, \\Yongchuan Niu, Bo Jiang, Yongjun Zhao, and Qiang Yan}
\affiliation{%
  \institution{TikTok Inc.}
  \country{}
}
\thanks{*Corresponding author: Jian Du. Email: {jian.du@tiktok.com}.}

\renewcommand{\shortauthors}{}

\begin{abstract}
In digital advertising, accurate measurement is essential for optimizing ad performance, requiring collaboration between advertisers and publishers to compute aggregate statistics—such as total conversions—while preserving user privacy. Traditional secure two-party computation methods allow joint computation on single-identifier data without revealing raw inputs, but they fall short when multidimensional matching is needed and leak the intersection size, exposing sensitive information to privacy attacks.

This paper tackles the challenging and practical problem of multi-identifier private user profile matching for privacy-preserving ad measurement, a cornerstone of modern advertising analytics. We introduce a comprehensive cryptographic framework leveraging reversed Oblivious Pseudorandom Functions (OPRF) and novel blind key rotation techniques to support secure matching across multiple identifiers. Our design prevents cross-identifier linkages and includes a differentially private mechanism to obfuscate intersection sizes, mitigating risks such as membership inference attacks.

We present a concrete construction of our protocol that achieves both strong privacy guarantees and high efficiency. It scales to large datasets, offering a practical and scalable solution for privacy-centric applications like secure ad conversion tracking. By combining rigorous cryptographic principles with differential privacy, our work addresses a critical need in the advertising industry, setting a new standard for privacy-preserving ad measurement frameworks.

\end{abstract}

\maketitle

\section{Introduction}

\noindent{\bf Private Ad in Practice:}
Digital advertising relies on cross-site and cross-app user tracking to link ad impressions stored by ad providers with user conversions recorded by advertisers, using third-party cookies or mobile device identifiers. This connection provides insights into ad performance and audience reach, enabling measurement of key metrics such as click-through rates, conversions, and return on ad spend (ROAS). By analyzing user behavior across platforms, marketers can further refine campaign strategies.

However, growing privacy concerns over tracking consumer behavior across ad providers and advertisers---driven by regulations such as GDPR and CCPA, along with industry changes such as third-party cookie deprecation and mobile identifier restrictions---are transforming the digital advertising landscape. In response, privacy-preserving technologies are becoming essential to balance effective ad measurement with robust data protection.

Ideally, if a unique identifier is used consistently across both the ad provider and the advertiser sides, cryptographic techniques that enable privacy-preserving set intersection can compute shared statistics without revealing sensitive data from either party, effectively preventing cross-platform tracking.
However, in practice, matching user identities across platforms is complex because users often have different identifiers, such as varying registration emails or phone numbers. To improve accuracy, multiple identifiers must be aligned based on specific rules. Techniques like the waterfall matching approach (also known as advanced matching~\cite{advanced_matching}), widely used in the advertising industry, prioritize identifier hierarchy to reduce false positives and enhance reliability. This method systematically resolves identity discrepancies by sequentially matching identifiers from the most reliable to the least reliable, ensuring more precise attribution. The underlying mathematics extends beyond the simple set intersection and is known as \textit{multidimensional intersection}, where multiple attributes are matched simultaneously to achieve context-aware identity resolution with higher accuracy. 

\subsection{Problem Characteristics and Requirements}

In Figure~\ref{fig:waterfall_match}, we provide a formal description of the functionality that we aim to compute privately.
This figure illustrates the process by which one party, \PB{}, learns the aggregated sum $c$ according to the logic of the waterfall matching\footnote{ 
Indeed, the waterfall matching can work for more than two identifiers. To simplify the presentation, we only consider two identifiers here.} for a multidimensional intersection. 
At first glance, the target problem illustrated in Figure~\ref{fig:waterfall_match} may appear similar to 
the private-match-for-compute protocols such as~\cite{google-psi,buddhavarapu2020private, ion2020deploying, eurocrypt/Pinkas0TY19,popets/ChandranGS22, garimella2021private}. 
However, we highlight two key differences that distinguish our target problem from theirs:

\noindent \textbf{1) Multidimensional Intersection.} 
For the case of a single ID, two databases are joined using exact matching on the common identifier, which is handled by most existing private set intersection protocols. 
However, when it comes to multiple IDs, it often results in many-to-many connections. 
For instance, some records in the databases may be matched by one identifier, e.g., email address, while other records are matched by a different identifier, e.g., phone number. 
The waterfall matching converts the many-to-many connections to one-to-one connections in a \textit{hierarchically iterative way}.
Specifically, the waterfall matching approach ranks identifiers in a predefined order.
A record is considered matched if one of its identifiers matches. 
To avoid many-to-many connections, the matched rows are removed before the matching procedure advances to the next identifier. Waterfall matching for multidimensional intersections is effective because:
\begin{enumerate}
    \item Each record is matched only once, preventing many-to-many connections.
    \item Higher-priority identifiers take precedence, reducing false positives and enhancing matching reliability.
\end{enumerate}

\noindent \textbf{2) Hiding Cross-ID Leakages.}
Leveraging multiple identifiers improves match rates, which is critical for the advertising business. However, a simple approach that runs a single-ID matching protocol separately for each identifier can unintentionally reveal additional information.

Figure~\ref{fig:waterfall_match} defines two types of cross-ID leakages: {\bf type-E} and {\bf type-X} (see Figure~\ref{fig:toy_cross_id_match} for an example). Type-E leakages occur when parties learn that a record is matched by multiple IDs, while type-X leakages occur when one party (say, \PA) learns that a record from \PB{} matches multiple records on \PA{}'s side.
For example, in Figure~\ref{fig:toy_cross_id_match}, the fourth row of the left table matches both the fourth and fifth rows of the right table. Type-X leakages are common in ad applications, where users may have multiple accounts on a publisher's platform registered with different identifiers, such as phone numbers and email addresses. 

Both type-E and type-X leakages increase the risk of membership leakage, as demonstrated by membership inference attacks described in the literature~\cite{Guo2022birds, Bo}. From a business perspective, these leakages allow either party to track user behaviour on the other side. For instance, the ad provider could learn if a user has converted on the advertiser's side, which could be particularly sensitive for industries like healthcare and finance. Addressing these leakages is crucial to ensure privacy while enabling effective matching.

\figbox{Functionality $\F{wfm}$}{
\small
On receiving identifiers $\ID^{1}_A, \ID^{2}_A$ from \PA{}, $\ID^{1}_B, \ID^{2}_B$ and payloads $T$ from \PB{}, this functionality do:
  \begin{enumerate}
    \item Match on the 1st identifier 
      \begin{align*}
      {J}^1_A = \{i: {\ID}^{1}_{A}[i] \in {\ID}^{1}_B\}, \quad {J}^1_B = \{j: {\ID}^1_B[j] \in {\ID}^{1}_A\}
      \end{align*}
    \item Unmatched records on the 2nd identifier
      \begin{align*}
        {\sf rID}^2_A = \{{\ID}^2_{A}[i]\}_{i \notin J^1_A}, \quad {\sf rID}^2_B = \{{\ID}^2_{B}[j]\}_{j \notin J^1_B}
      \end{align*}
    \item Match on the 2nd identifier
      \begin{align*}
      &{J}^2_A = \left\{i: {\sf ID}^{2}_{A}[i] \in {\sf rID}^2_B \right\} \setminus J^1_A\\
      &{J}^2_B = \left\{j: {\sf ID}^{2}_{B}[j] \in {\sf rID}^2_A \right\} \setminus J^2_B
      \end{align*}
    \item Give \PA{} the sizes $s_1 = |J^1_A|$ and $s_2 = |J^2_B|$.
    \item Give \PB{} the sum $c = \sum_{i\in {J}^1_B\cup {J}^2_B} T[i]$ and $s_1, s_2$.
  \end{enumerate}
}{
The waterfall matching and sum functionality\label{fig:waterfall_match}.
Note that the matches on the 2nd ID column can also be given as 
$\{i: {\sf rID}^2_A[i] \in {\sf rID}^2_B \}$ and $\{j: {\sf rID}^2_B[j] \in {\sf rID}^2_A \}$, which needs an extra
mapping to a global index in the payload column.
}

\begin{figure}[t]
\centering
\begin{tikzpicture}
    \definecolor{lightblue}{RGB}{225, 237, 255}
    \definecolor{lightyellow}{RGB}{255, 249, 227}
    \definecolor{darkorange}{RGB}{237, 125, 49}
    \definecolor{darkgreen}{RGB}{0, 128, 0}
     \renewcommand{\arraystretch}{1.3}

    \node (tableA) at (0,0) {
        \begin{tabular}{|>{\columncolor{lightblue}}c|>{\columncolor{lightblue}}c|}
            \hline
            $\ID_{A}^1$ & $\ID_A^2$ \\ \hline
            d & ! \\ \hline
            a & \# \\ \hline
            f & \& \\ \hline
            g & @ \\ \hline
        \end{tabular}
    };

    \node (tableB) at (5, 0) {
        \begin{tabular}{|>{\columncolor{lightyellow}}c|>{\columncolor{lightyellow}}c|}
            \hline
            $\ID_B^1$ & $\ID_B^2$ \\ \hline
            b & \^{} \\ \hline
            a & \# \\ \hline
            c & \& \\ \hline
            g & NA \\ \hline
            NA & @ \\ \hline
        \end{tabular}
    };

    \draw [<->, thick, darkorange] ([yshift=0.0ex]tableA.east) -- ([yshift=1.5ex]tableB.west) node[midway, above] {type-E};
    \draw [<-, thick, darkgreen, dashed] ([yshift=-7.0ex]tableA.east) -- ([yshift=-5ex]tableB.west) node[midway, above, sloped] {type-X};
    \draw [<-, thick, darkgreen, dashed] ([yshift=-7.0ex]tableA.east) -- ([yshift=-8ex]tableB.west);
\end{tikzpicture}

\caption{Example of the type-E cross-ID leakage and the type-X cross-ID leakage. `NA' indicates a missing value.\label{fig:toy_cross_id_match}}
\end{figure}
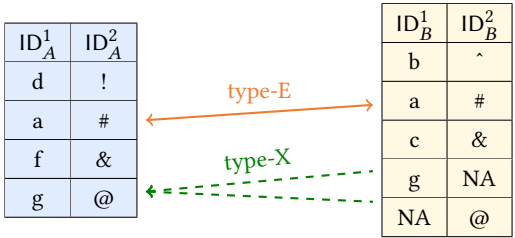
\section{Technical Overview} 

\subsection{Distributed Evaluation of Reversed and Blindly Updatable Pseudorandom Function}
We first recap the recipe for building a private set intersection protocol.
Suppose \PA{} inputs set $X$ and \PB{} inputs set $Y$.
They invoke an Oblivious Pseudorandom Function (OPRF) protocol, which 
gives the PRF evaluations $\{F_k(x) \mid x \in X\}$ to \PA{} and gives the key $k$ to \PB{}.
This directly allows obtaining the intersection $X \cap Y$ by having \PB{} send $\{F_k(y) \mid y \in Y\}$ to \PA{}.
Since the interest is in evaluating a function $f$ over the intersection, rather than the intersection itself,
it is natural to perform the OPRF in a reversed manner. 
For example, \PB{} obtains $\{F_k(x) \mid x \in X\},$ the PRF evaluations on \PA{}'s input. 
By performing the reversed OPRF twice and exchanging the PRFs,
\PA{} and \PB{} can obtain the set cardinality $f(X \cap Y) = |X \cap Y|$ without revealing the intersection.
This {\it reversed OPRF} implicitly demands a {\it distributed OPRF} evaluation, i.e., the corresponding PRF key is secretly shared between the two parties. 


We demonstrate that the reversed OPRF already enables efficient row deletion required by the waterfall matching logic. Specifically, \PA{} can locally remove matched rows from ${F_k(y) \mid y \in Y}$ before exchanging PRF evaluations with \PB{}. Since the key remains unknown to both parties, \PB{} is oblivious to which rows have been deleted. This private row deletion inherently mitigates type-E leakages.

To address type-X leakages, we introduce a novel mechanism for blind PRF key updates. Given an evaluation $F_k(x)$, \PB{} interacts with \PA{} to obtain $F_{k'}(x)$—an evaluation under a new shared key—without revealing the linkage between $F_k(x)$ and $x$. Importantly, this ensures that even if \PA{} knows the input $x$, it remains unaware of the key transition, thus preserving input privacy across key updates. Our method represents a significant advancement in leakage resilience by seamlessly integrating key blinding into private matching protocols.


Finally, we formalize the desired OPRF functionality in Figure~\ref{fig:recipe}, defining it as a {\it Distributed, Blindly Updatable, and Reversed Pseudorandom Function ($\PRF$)}. This new $\F{\PRF}$ construction enables a privacy-preserving multidimensional intersection for implementing a waterfall matching protocol with robust leakage resistance.
Our blind key update mechanism addresses type-X leakages by fully decoupling key transitions from input knowledge. This ensures that even if an input is partially known to \PA{}, the linkage between evaluations under different keys remains hidden, preserving privacy across updates. Additionally, the reversed OPRF design allows \PA{} to efficiently delete matched rows from ${F_k(y) \mid y \in Y}$ without exposing deletion patterns to \PB{}, inherently mitigating type-E leakages.
By integrating these capabilities, our protocol delivers privacy guarantees against adversaries seeking to exploit linkage or key-correlated information, providing a comprehensive defense against both type-E and type-X leakages. 

\figbox{Functionality $\F{\PRF}$}{
\small
There are two parties \PA{} and \PB{}. 

{Initialization:} Upon receiving ${\bf Init}.({\sf sid}, k_A)$ from \PA{}, and  
${\bf Init}.(\sid, k_B)$ from \PB{}, 
initialize an empty table $T_{\sid}$, 
and record $(\sid, k_A \oplus k_B)$.
Ignore all subsequent ${\sf init}$ queries for the session ID $\sid$.

{Evaluation:} On ${\bf Send}.(\sid, x_1, \cdots, x_n)$ from $\PA$ (sender), and ${\bf Recv}.(\sid, n)$ from \PB{} (receiver) do:
\begin{itemize}
  \item Set $T_{\sid}[x_i] \usample \{0, 1\}^\lambda$ if any empty entry.
  \item Give $T_{\sid}[x_1], \cdots, T_{\sid}[x_n]$ to \PB{}.
\end{itemize}
The roles can be flipped, \ie, \PB{} acts as the sender and \PA{} acts as the receiver.

{Update:} Upon receiving ${\bf Update}.({\sid}, {\sid}', z_1, \cdots, z_m)$ from \PB{},
and ${\bf Update}.(\sid, \sid', m)$ from \PA{} do:
\begin{itemize}
  \item Set $z'_j \usample \{0, 1\}^\lambda$ for $j = 1,2,\cdots, m$.
  \item For each $z_j$, find an entry $x_j$ in $T_{\sid}$ such that $T_{\sid}[x_j] = z_j$. 
        If not any then set $x_j = \bot$.
  \item If $x_j\neq \bot$ and the entry $T_{\sid'}[x_j]$ is not empty, 
        then overwrite $z'_j \leftarrow T_{\sid'}[x_j]$.
        Otherwise if $x_j\neq \bot$ and the entry $T_{\sid'}[x_j]$ is empty, 
        then set the entry $T_{\sid'}[x_j] \leftarrow z'_j$.
  \item Finally return $z'_1, \cdots, z'_m$ to \PB{}.
\end{itemize}
The roles can be flipped too, \ie, \PB{} acts as the assistant and \PA{} acts as the updater.
}{The functionality of $\PRF$. \label{fig:recipe}}

\subsection{Enhancing Privacy via Differential Privacy}
Recent research highlights that size-revealing (single-ID) private set intersection (PSI) protocols are susceptible to membership inference attacks, posing significant privacy risks in advertising contexts~\cite{Guo2022birds, Bo}. In these attacks, an ad publisher can infer which users on its platform viewed an ad but did not complete a purchase with the advertiser. This compromises sensitive business data, as the advertiser may not wish to disclose granular user conversion behavior. Such leakages not only undermine the advertiser’s competitive edge but also raise potential compliance concerns under privacy regulations such as the GDPR, which mandates strict limits on personal data processing and requires minimizing data exposure between parties. Addressing these vulnerabilities is critical to ensuring privacy-preserving advertising practices that align with regulatory obligations and protect proprietary business insights.

Multidimensional intersections are more vulnerable to membership inference attacks. Particularly, we conduct the membership inference attacks~\cite{Guo2022birds, Bo} against the size-revealing and 
multi-ID protocol~\cite{multi_key_meta}. 
The results show that the success rates doubled (\eg, from 2.5\% to 6.6\%) when the same attacks were carried out over a dataset with 4 IDs. See Figure~\ref{fig:single_vs_mult_id_with_DP} in \S\ref{sec:experiments}.


To mitigate intersection size leakages, we adopt a strategy based on differentially private (DP) mechanisms. Specifically, we introduce a novel {\it dual-sided DP} mechanism that enables both parties to independently sample and inject dummy identifiers into their datasets, $\ID_A$ and $\ID_B$, creating two augmented sets, $\widetilde{\ID}_A$ and $\widetilde{\ID}_B$, respectively.

Input privacy is preserved by the private waterfall matching protocol, while the intersection size privacy is guaranteed by ensuring that the distribution of $|\widetilde{\ID}_A \cap \widetilde{\ID}_B|$ is differentially private with respect to the true size $|\ID_A \cap \ID_B|$. This mechanism obfuscates the intersection size, making it indistinguishable within a carefully controlled privacy budget.

The number of dummy identifiers added depends on the size of the original data set and the desired privacy parameters. By leveraging our $\PRF$ construction, which enables independent processing of dimensions, we achieve a tighter DP analysis. This design allows us to apply parallel composition across multiple dimensions, significantly reducing the exponential overhead of dummy identifiers required under a sequential composition approach. Consequently, our method scales more efficiently while maintaining rigorous privacy guarantees, offering a substantial improvement in the practicality of privacy-preserving multi-dimensional set operations.

\noindent{\textbf{On DP to Multi-Dimension PSI}.}
A naive approach of applying DP to PSI protocols dimension-by-dimension may appear feasible at first glance. 
However, while DP mechanisms can theoretically mitigate privacy risks posed by type-E and type-X leakages, this approach incurs exponentially increasing overheads. 
This stems from the fact that the intersection size can vary by up to the full dataset size, as demonstrated by the scenario where all records of one party could potentially match different identifiers under type-E/X leakage conditions. 
In contrast, our waterfall matching protocol \textbf{eliminates} these leakage vectors entirely, enabling a significantly reduced overhead when integrated with DP mechanisms.

\noindent{\bf On the Input Validations.}
The intersection sizes are effectively randomized by our {\it dual-sided DP} mechanism, provided both parties correctly sample and add dummies to their input sets. However, an adversary could bypass this protection and reveal the exact intersection size by simply skipping the sampling step.
To prevent this ``simple attack,'' we employ secure hardware, such as Intel SGX, for code attestation. Each party uses its own secure hardware to verify that the (encrypted) messages they send are generated by the intended code, including the correct dummy sampling. Since the dummy set is common information, it can be hardcoded within the attested code. Notably, confidentiality of the hardware is unnecessary in this setting, as it is managed and executed by the same party.
\subsection{Our Contributions}
Our contributions are summarized as follows:

\begin{itemize}
\item \textbf{A general framework for private waterfall matching without cross-ID leakage.}

We propose a novel framework to achieve the waterfall matching functionality while eliminating cross-ID leakages. Central to our approach is a functionality we define as the \textit{Distributed, Updatable, and Reversed Pseudorandom Function} ($\F{\PRF}$), described in Figure~\ref{fig:recipe}. We present two constructions for $\F{\PRF}$ using hashed Diffie-Hellman (DH) and Yao’s Garbled Circuit~\cite{yao1986generate}.

Preventing cross-ID leakage is critical because it allows the use of lighter, more efficient differential privacy (DP) mechanisms to enhance privacy. Applying DP to protocols with cross-ID leakage would result in significantly higher overhead, as privacy loss grows exponentially with the number of IDs. Our solution effectively decouples privacy loss from cross-ID matches, leading to more scalable and practical implementations.
\item \textbf{A tightly bounded DP mechanism for protecting intersection sizes.}
We design a differentially private mechanism where both parties independently sample and add dummy identifiers to their input sets. This protects the intersection size with minimal overhead.

We present two analyses to reduce dummy size, improving efficiency. The first uses $\F{\PRF}$ to decompose cross-ID leakage, allowing parallel composition with linear dummy growth relative to the number of identifiers, compared to the exponential growth in sequential composition. The second uses convolution techniques from~\cite{koskela2020tight} to minimize dummies across multiple waterfall matching executions, such as aggregating conversion data from two ad campaigns.

\item \textbf{Two variants and an optimized implementation.}
We extend our protocol to support both \PA{} and \PB{} providing payloads for aggregation and enable general functions beyond summation by sharing matched homomorphic encryption (HE) ciphertexts for secret sharing.

We implement the protocol using DH-based $\PRF$ and a lattice-based HE scheme, with HE optimizations that handle two million records across three ID columns in about 2 minutes (8 threads) over a 100Mbps connection. Our DP-enhanced protocol achieves practical performance, and we plan to release our implementation publicly.
\end{itemize}
\subsection{Related Works}
In Private Set Intersection (PSI) two parties, \PA{} with a set $X$ and $\PB{}$ with a set $Y$ , securely compute the intersection $X \cap Y$, without leaking
the information of the items that are not in the intersection has been studied extensively in a long sequence of works~\cite{chen2017fast,pinkas2014faster,kolesnikov2016efficient,freedman2004efficient,google-psi,PS3I,multi_key_meta,eurocrypt/Pinkas0TY19,popets/ChandranGS22,eurocrypt/Pinkas0TY19,cans/KarakocK20, pkc/0005DP21,crypto/GhoshS19,access/ZhangCL21, ion2020deploying}. 

In many settings, the goal is to compute some function $f$ over the intersection set, i.e., $f(X \cap Y)$, rather than knowing the intersection itself~\cite{DBLP:conf/scn/CiampiO18, google-psi, DBLP:conf/cans/CristofaroGT12,DBLP:conf/isw/DebnathD15}.
PSI-cardinality is one example of such an the two parties are limited to learning only the cardinality (or size) of the intersection~\cite{DBLP:conf/cans/CristofaroGT12,DBLP:conf/isw/DebnathD15}.
The private intersection-sum functionality introduced by Ion \etal~\cite{google-psi} is another example where one of the input sets has integer values associated with the elements in the set  and the two parties aggregate of the integer values associated with the intersection set.
The circuit-based PSI protocols such as~\cite{eurocrypt/Pinkas0TY19,popets/ChandranGS22,DBLP:conf/ndss/HuangEK12} enable 
the computation of arbitrary symmetric functions securely over the intersection. 
To the best of our knowledge, most of the existing PSI protocols are designed for the single identifier.

Private-ID~\cite{multi_key_meta} is the only publicly available solution for multi-ID matching that comes with membership leakage.
More precisely, the participants in \cite{multi_key_meta} can know the intersection size of any subset of IDs, leading to an exponential leakage of membership information.


Kacsmar \etal~\cite{eurosp/KacsmarKLNSSBSO20} employ DP mechanisms for the cardinality function $f(X\cap Y) = |X \cap Y|$  in a client-server scenario in where
only the client can know the result.  Consequently, only central DP is necessary for their setting.
However, in our context, where each party is inquisitive about membership and there is a potential for membership attacks, the application of central DP is not suitable. Instead, a distributed DP approach is required.

\noindent{\bf Oblivious PRF.} In the literature, the oblivious evaluation of a PRF is commonly defined as a two-party protocol where \PA{} provides a key $k$ and the other party \PB{} provides an input $x$. At the end of the protocol execution, \PB{} obtains $F_k(x)$. There are many constructions for oblivious PRF protocols, such as~\cite{DBLP:conf/focs/NaorR97,DBLP:conf/pkc/DodisY05,DBLP:conf/ccs/KolesnikovKRT16}.
Moreover, we can modify an OPRF protocol to a key-distributed and output-reversed counterpart, i.e., the key is distributed between the two parties and \PA{} obtains the evaluation, using homomorphic encryption as described by~\cite{google-psi}. The double-hash approach from~\cite{DBLP:conf/cans/CristofaroGT12} is a reversed OPRF but does not provide key rotation capability.

Many OPRF protocols based on the Dodis-Yampolskiy (DY) PRF can support key rotation (i.e., updating the PRF key). However, it seems to require nontrivial efforts to support blind key rotation when the DY-PRF key is distributed between two parties. We refer to the survey by Casacuberta \etal~\cite{DBLP:conf/eurosp/CasacubertaHL22} for more details on recent oblivious PRF constructions.

\noindent{\bf Circuit-based Solutions.}
One can also leverage the technique called circuit-based PSI (\cPSI), such as~\cite{eurocrypt/Pinkas0TY19,popets/ChandranGS22,DBLP:conf/ndss/HuangEK12},
which can eliminate the leakage of the intersection sizes.
However, the modification of the current single-ID \cPSI{} protocols to the multi-ID setting involves non-trivial work. 
A \cPSI{} protocol is basically an asymmetric set membership testing. 
That is it allows $\PB{}$ to query whether its ID match the ones own by $\PA{}$.
At the end of the protocol execution,
\PB{} obtains a secretly shared testing bits 
$b_i = 1$ if $\id^{1}_{B, i} \in \ID^{1}_A$, or $b_i = 0$ otherwise.
One of the difficulties of adopting the \cPSI{} to the multi-ID setting is to obliviously delete records that already matched from the previous IDs. 
That is because the position of each ID is computed using a hash function. 
For instance, the position of the identifier ``a@123'' 
is computed via $i = {\sf Hash}\mbox{("a@123")}$, and then it is assigned to the $i$-th position in the query vector.
It renders an inconsistent order of result bits when performing the queries on a record of different IDs.
Communication extensive cryptographic protocols (e.g., oblivious permutation) are needed to align the order of the of testing bits (and the payloads). 

In the ad applications, the intersection size can be significantly smaller than the input sets, \ie, $|X \cap Y| \ll {\rm min}(|X|, |Y|)$.
For instance, a typical ratio is less than 1\%. 
In other words, to evaluate a downstream function over the intersection using  fully private techniques might introduce a large overhead on the downstream computation.

\section{Preliminaries}

\subsection{Definitions}

\begin{definition}(Pseudorandom Functions~\cite{DBLP:books/crc/KatzLindell2014})
\label{def:prf}
Let $F: \{0, 1\}^* \times \{0, 1\}^* \mapsto \{0, 1\}^*$ be an efficient, length-preserving, keyed function. 
We say $F$ is a pseudorandom function if for all probabilistic polynomial-time distinguishers $\mathcal{D}$, there exists a negligible function negl such that:
$\Pr[\mathcal{D}^{F_k(\cdot)}(1^n) = 1] - \Pr[\mathcal{D}^{f_n(\cdot)}(1^n) = 1] \le {\sf negl}(n)$,
where $k\usample  \{0, 1\}^n$ is chosen uniformly at random and $f_n$ is chosen uniformly at random from the set of functions mapping $n$-bit strings to $n$-bit strings.
\end{definition}

\noindent{\bf Honest-but-Curious Privacy.}
We recap the privacy definition from~\cite{cu/Goldreich2004,DBLP:journals/joc/Canetti00}. 
Let $F: \bin^* \times \bin^* \mapsto \bin^*$ be a deterministic functionality
where $F_0(x_0, x_1)$ (resp. $F_1(x_0, x_1)$) denotes the 1st element
(resp. the 2nd) of $F(x_0, x_1)$, and let $\Pi$ be a two-party protocol
for computing $F$.
The view of $\PA{}$ (resp. $\PB$) during an execution of $\Pi$ on $(x_0, x_1)$ is denoted 
$\VIEW^{\Pi}_A(x_0, x_1)$ (resp. $\VIEW^{\Pi}_B(x_0, x_1))$.
\begin{definition} (Honest-but-Curious Privacy)
\label{def:HBC}
For a function $F$, we say
that $\Pi$ privately computes $F$ if there exist  probabilistic polynomial time algorithms,
denoted $\SIM_0$ and $\SIM_1$, such that
\begin{align*}
 \{\SIM_0(x_0,F_0(x_0, x_1))\}_{x_0,x_1} \overset{c}{\equiv} \{\VIEW^{\Pi}_A(x_0, x_1)\}_{x_0,x_1} \\
 \{\SIM_1(x_1,F_1(x_0, x_1))\}_{x_0,x_1} \overset{c}{\equiv} \{\VIEW^{\Pi}_B(x_0, x_1)\}_{x_0,x_1},
\end{align*}
where $\overset{c}{\equiv}$ denotes computational indistinguishability.
This definition states that the views of the parties can be properly constructed by a polynomial time algorithm given
the party's input and output solely. 
Also, the parties here are semi-honest and the view is therefore exactly according to the protocol specification.

\end{definition}

\noindent{\bf Differential Privacy.}
We recall a definition of approximate differential privacy for interactive two-party protocols, following~\cite{mcgregor2010limits}.
Let $\Sigma$ be a finite alphabet and for strings $x, y \in \Sigma^n$, let $|x - y|_H$ denote the Hamming distance between $x$ and $y$.
\begin{definition}($(\epsilon, \delta)$-Differential Privacy)  
\label{def:dp}
A mechanism $\mathcal{M}$ on $\Sigma^n$ is a family of probability distributions $\{\mu_x : x \in \Sigma^n\}$ on $\mathcal{R}$.
The mechanism is $(\epsilon, \delta)$-differentially private if for every $x$ and $x'$ such that $|x - x'|_H = 1$ and every measurable subset $S \subset \mathcal{R}$ we have
  \[
    \mu_x(S) \le \exp(\epsilon) \mu_{x'}(S) + \delta.
  \]
\end{definition}

The definition of differential privacy naturally extends to interactive protocols, by requiring that the views of all parties be differentially private in respect to other parties' inputs. 
More specifically, let $\VIEW^{A}_{\Pi}(x, y)$ be the joint probability distribution over $x$, the transcript of the protocol $\Pi$, private
randomness of \PA{}, where the probability space is private randomness of both parties. 
For each $x$, $\VIEW^{A}_{\Pi}(x, y)$ is a mechanism over the $y$'s. Let $\VIEW^{B}_{\Pi}(x, y)$ be similarly defined view of \PB{} whose input is $y$.
\begin{definition}\label{def:2pc_dp} (Differential privacy for two-party protocols)
We say that a protocol $\Pi$ enables $(\epsilon, \delta)$-differential privacy if the mechanism 
$\VIEW^{A}_{\Pi}(x, y)$ is $(\epsilon, \delta)$-differentially private for all values of $x$,
and same holds for $\VIEW^{B}_{\Pi}(x, y)$ for all values $y$.
\end{definition}

\subsection{Security Model}
We chose to target the security against honest-but-curious, \ie semi-honest adversaries, due to 
the strong efficiency and monetary requirements of ad conversion applications.
The semi-honest model provides strong privacy protections against data breaches on either side since semi-honest protocols leak nothing beyond the prescribed protocol output.

We protect the intersection size by having each party independently sample and add dummy identifiers from a common set of dummies.
Even under the semi-honest setting, all parties should follow the protocol descriptions, e.g., sampling dummy identifiers from the proper set.
However, there is a prisoner's dilemma where an adversary can learn the exact intersection size if he/she simply skips the sampling step, and the honest party cannot detect this behavior. To guarantee that each party correctly performs the DP mechanism, we rely on secure hardware such as Intel SGX to provide code attestation.
In brief, the outputs generated by the hardware are indeed from the code that it attested to. 
Note that we \textbf{do not} assume confidentiality for the secure hardware since, in our case, a party uses his/her own hardware to convince the other party that he/she has executed a public piece of code properly. 
This practice of using secure hardware to achieve code attestation is also employed by other works such as~\cite{DBLP:conf/sp/0001RCGR020}.

\subsection{Additively Homomorphic Encryption (AHE)}
An AHE scheme $\mathcal{HE} = ({\sf Gen}, {\sf Enc}, {\sf Dec}, {\sf Sum}, {\sf Refresh})$
consists of 5 algorithms.
${\sf Gen}$ outputs a public-private key pair $(\pk, \sk)$, and specifies a message space $\mathcal{M}$.
Given the public key $\pk$ and a plaintext message $m \in \mathcal{M}$, one can compute a ciphertext $\AEnc(m)$, an encryption of $m$ under $\pk$.
Given the secret key $\sk$ and a ciphertext $C = {\sf Enc}(m)$, one can run ${\sf Dec}(C)$ to recover the plaintext $m$.

Given the public key $\pk$ and a set of ciphertexts $\{C_i\}$ encrypting messages
$\{m_i\}$, one can homomorphically compute a ciphertext encrypting the sum of the underlying messages, which we denote 
for ease of exposition as: ${\sf Enc}(\sum_i m_i) = {\sf Sum}(\{C_i\}_i)$.
We will also use the property that one can randomize ciphertexts using a randomized procedure denoted as ${\sf Refresh}$. 
The two distributions 
${\sf Refresh}({\sf Sum}(\{{\sf Enc}(m_i)\}_i))$ and 
${\sf Refresh}({\sf Enc}(\sum_i m_i))$
are statistically close, even against an adversary that holds the decryption key $\sk$.

We can instantiate the AHE scheme by the Brakerski/Fan-Vercauteren scheme (BFV) scheme~\cite{DBLP:conf/crypto/Brakerski12,DBLP:journals/iacr/FanV12} or the Paillier scheme~\cite{DBLP:conf/eurocrypt/Paillier99}.
Particularly, the current work~\cite{google-psi} used the Paillier scheme while we present some optimizations for the BFV scheme.

\section{Proposed Protocol for Waterfall Matching}

\begin{figure*}
\centering
\includegraphics[width=\textwidth]{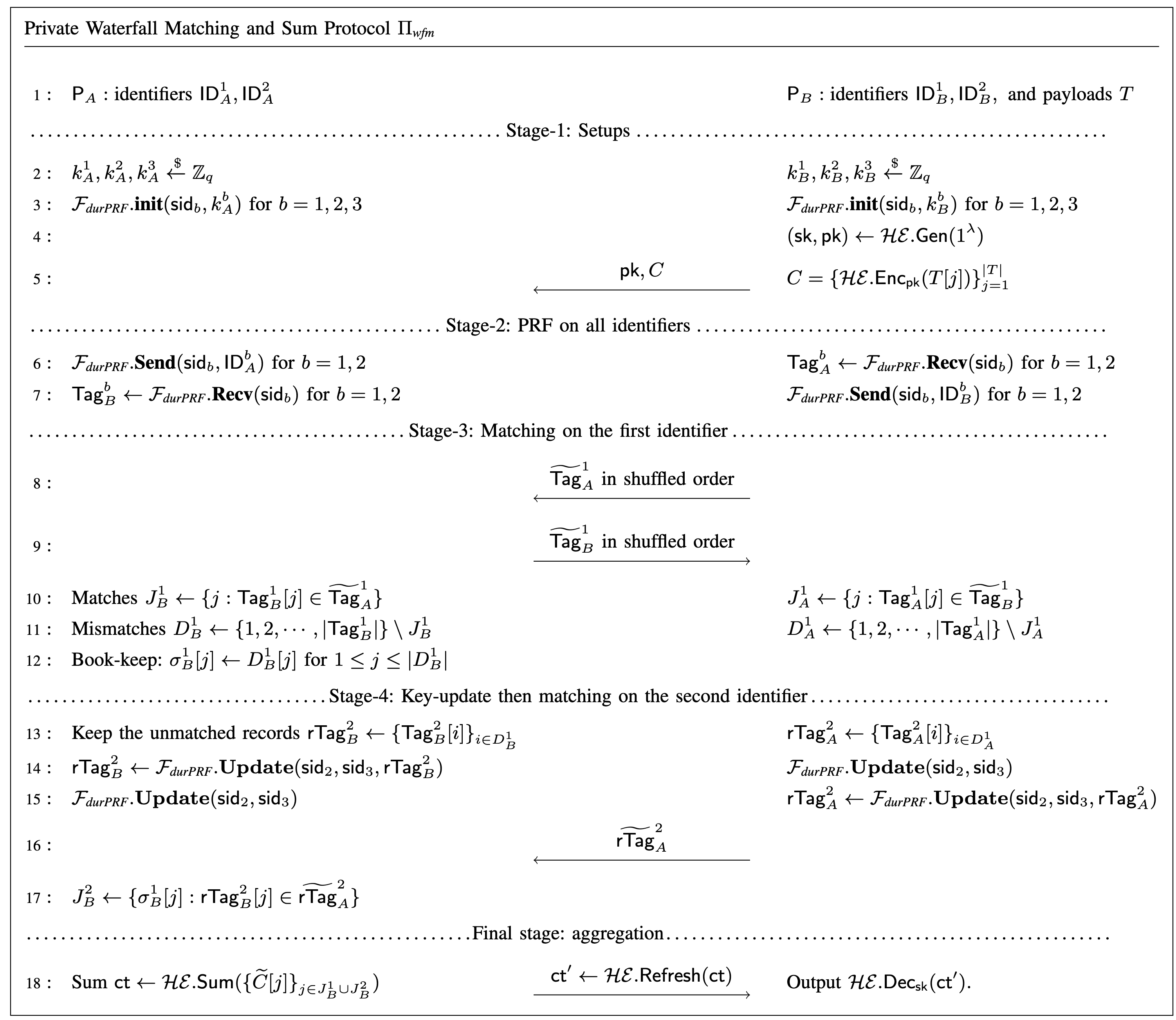}
  \caption{Our private waterfall matching protocol under $\F{\PRF}$-hybrid model. 
  $\mathcal{HE}$ is an AHE scheme.
  \label{fig:proposed_ddh}
}
\end{figure*}

\subsection{Private Waterfall Matching with Sum Protocol from $\F{\PRF}$}

The concrete construction of our private waterfall matching protocol is given in Figure~\ref{fig:proposed_ddh}.
Our protocol operates in five stages. 


\begin{enumerate}
\item In the first stage, \PB{} sends homomorphic encryption of \PB{}'s payloads to \PA{} along with his public key.
\item The two parties first send the ``random tags'' of their IDs using multiple $\PRF$.
For instance, on the $b$-th ID column,
$\PB$ sends $\mathcal{F}_{\PRF}.{\bf Send}(\sid_b, \ID_B^b)$ 
and $\PA$ receives ${\sf Tag}^b_B  \leftarrow \mathcal{F}_{\PRF}.{\bf Recv}(\sid_b)$.
That is \PA{} obtains the PRF evaluations on \PB{}'s $b$-th ID column under a shared key $k_b$.
On the other hand, \PB{} receives the PRF evaluations on \PA{}'s IDs under the corresponding key, denoted as ${\sf Tag}^b_A$.

  \item By exchanging the tags on the 1st ID column, \PA{} and \PB{} can obtain the intersection in a reversed way.
    For instance, after receiving  $\widetilde{{\sf Tag}}^1_A$ in a shuffled order from \PB{}, 
    \PA{} knows the matching positions 
    $
    J_B^1 = \{i: {\sf Tag}^1_B[i] \in \widetilde{{\sf Tag}}^1_A\}
    $
    on the 1st ID column, but \PA{} does not learn the exact matched row index in \PA{}'s dataset due to the shuffling.
    Similarly \PB{} knows $J_A^1 = \{i: {\sf Tag}^1_A[i] \in \widetilde{{\sf Tag}}^1_B\}$.
    
  \item Moreover, to avoid the type-E cross-ID leakage,
     \PA{} and \PB{} can first remove the matched rows  from the received tags of the 2nd ID column.
     Particularly, $\PA$ prepares ${\sf rTag}^2_B = \{{\sf Tag}^2_B[i] = F_{k_2}(\ID^B[i])\}_{i \notin J_B^1}$
     and $\PB$ prepares ${\sf rTag}^2_A = \{{\sf Tag}^2_A[i] = F_{k_2}(\ID^2_A[i])\}_{i \notin J_A^1}$.
    Then they perform the matching on the remaining tags.
    For instance, on receiving $\widetilde{{\sf rTag}}^2_A$ from \PB{},
    \PA{} can know the matching positions due to the 2nd ID column: $J_B^2 = \{\sigma(i): {\sf rTag}^2_B[i] \in \widetilde{{\sf rTag}}^2_A\}$
    where $\sigma$ is an index mapping from $|{\sf rTag}^2_A| \mapsto |\ID^1_A|$ that maps the matching positions in ${\sf rTag}^2_A$ to their global positions in $\ID^1_A$.
    \textbf{The type-E cross-ID leakage is avoided from the sense that $J_B^1\cap J_B^2 = \emptyset$.}

    However, it might introduce the type-X cross-ID leakage.
    For instance, when $\exists i \in J^1_B \Rightarrow {\sf Tag}^2_B[i] \in \widetilde{{\sf rTag}}^2_A$, it reveals to $\PA$ that
    a record $i$ of \PB{} is matched by two records of \PA{}, one due to the 1st ID, and the other is due to the 2nd ID.
    To mitigate the type-X cross-ID leakage, we further leverage the blind-update command as follows
     \begin{align*}
       \PB: {\sf rTag}^2_A \leftarrow \mathcal{F}_{\PRF}.{\bf Update}(\sid_2, \sid'_2, {\sf rTag}^2_A)\\
       \PA: {\sf rTag}^2_B \leftarrow \mathcal{F}_{\PRF}.{\bf Update}(\sid_2, \sid'_2, {\sf rTag}^2_B)
     \end{align*}
    \textbf{so that 
    ${\sf Tag}^2_B \cap {{\sf rTag}}^2_A = \emptyset$ 
    and ${\sf Tag}^2_A \cap {{\sf rTag}}^2_B = \emptyset$ due to the distinct PRF keys. Thus the type-X leakage is also avoided.}

  \item 
        \PA{} can then aggregate the HE ciphertexts using the matching lists $J^1_B$ and $J^2_B$.
\end{enumerate}

\begin{theorem}
\label{thm:proposed_ddh_security}
If the AHE scheme $\mathcal{HE}$ is semantic secure then the protocol $\Prot{wmf}$ in Figure~\ref{fig:proposed_ddh} securely realizes
the $\F{wmf}$ functionality of Figure~\ref{fig:waterfall_match} 
against semi-honest adversaries under $\F{\PRF}$-hybrid.
\end{theorem}

\begin{proof}
Let $\Adv$ be a polynomial-time adversary that may corrupt $\PA$ and $\PB$.  
We construct a simulator $\Sim$ with access to functionality $\F{\PRF}$ that runs $\Adv$ as a subroutine. 
Moreover the $\Sim$ is also given the intersection sizes $s_1$ and $s_2$. 
We write $\mathcal{O}$ to denote the domain of the PRF, and
write $n_A$ and $n_B$ to denote the database sizes of $\PA$ and $\PB$, respectively.


\noindent{\textbf{Corrupted $\PA{}$ with a honest $\PB{}$.}} 
\begin{enumerate}
     \item Let $(k^1_A, k^2_A, k^3_A)$ be the keys  $\Adv$ sends to $\F{\PRF}$.\textbf{Init}. $\Sim$ sends an AHE public key, and $n_B$ AHE ciphertexts of zero to $\Adv$ on behalf of $\PB$.
     \item Let $\sid_b$ ($b = 1, 2, 3$) be the session id $\Adv$ sends to $\F{\PRF}$.\textbf{Recv}. $\Sim$ sends $n_B$ random elements of $\mathcal{O}$ to $\Adv$ on behalf of $\F{\PRF}$.
     \item $\Sim$ sends $n_A$ random elements from $\mathcal{O}$ to $\Adv$ on behalf of $\PB{}$.
     \item Let $\sid_2, \sid_3$ be the session ids $\Adv$ sends to $\F{\PRF}$.\textbf{Update}. $\Sim$ samples $n_B - s_1$ random elements (designated by $R$) from $\mathcal{O}$ to $\Adv$ on behalf of $\F{\PRF}$.
     \item $\Sim$ samples $n_A - s_1$ random elements of $\mathcal{O}$ and randomly replaces $s_2$ of them by a subset of $R$.
     Then $\Sim$ sends them to $\Adv$ on behalf of $\PB$.
\end{enumerate}
It is not hard to see that the simulation is computationally close to an execution of $\Prot{wfm}$ in the $\F{\PRF}$-hybrid world.
In particular, the simulation relies on the fact that messages sent by $\PB$ are either AHE ciphertexts or PRF outputs, and the matching positions (e.g. $J^1_B, J^2_B$) are randomly shuffled.

The simulation for $\PB$'s view is similar except that we additionally simulate an AHE of aggregation result in Step 18 where the privacy is guaranteed by the $\AFresh$ function.
\end{proof}



\subsubsection{AHE Instantiation from (R)LWE}
Learning-With-Errors (LWE)-based encryption schemes are usually quite computationally efficient compared to schemes like Paillier and DGK, since they do not involve expensive modular-exponentiation operations to encrypt and decrypt. 
However, the batch of LWE ciphertexts of the payloads can be large in volume. 
To further reduce the communication overhead, we can first transfer a batch of $N$ (\eg, $N = 4096$) payloads inside a Ring LWE ciphertext instead of a batch of LWE ciphertexts. 
This reduces the communication by a factor of $O(N)$.
Then, the ciphertext receiver (\eg, \PA{} in our case) can unpack the encrypted payloads from the received Ring LWE ciphertext, \ie, one LWE ciphertext for each payload value, by performing some inexpensive local operations.
The aggregation is then performed over the LWE ciphertexts.

We explicitly export the Ring LWE as an encapsulation layer of the LWE ciphertexts to be more efficient and to offer more flexibility.
Particularly, for the efficiency, 
the sum step over LWE ciphertexts is about $2\times$ faster than the corresponding homomorphic additions over the Ring LWE ciphertexts.
For the flexibility, when comes to the secret sharing variant described in the following section, 
\PA{} can return a "packed" version of the matched LWE ciphertexts $\{C[j]\}_{j\in J^1_B \cup J^2_B}$ to \PB{} using the technique from~\cite{DBLP:conf/acns/ChenDKS21}.

\subsection{Variants}
We considered variants of the waterfall matching protocol that modify the functionality and offer important flexibility.

\subsubsection{Multiple Payloads} 
Our protocol generalizes easily to cases where each row in \PB{}'s dataset has multiple types of payloads (\eg amount spent, and label of conversion), and parties wish to learn the sum of each type of payload.
This is handled by simply encrypting each column separately using the AHE, and creating one sum for each type of payload.

\subsubsection{Secret Sharing the Matched Payloads} 
One variant is to let \PA{} reshare the matched HE payloads to \PB{} in an additively secret shared form, rather than computing the aggregated sum.
This variant allows more complicated functions beyond the aggregated sum to be evaluated on the intersection set.
To achieve this, \PA{} now sends $\AFresh\left(\AEnc(C[j] + r_j)\right)$ a batch of AHE ciphertexts to \PB{} in Step 17.
for $j \in {J}^1_B \cup {J}^2_B$. The masking values $\{r_j\}$ are sampled uniformly from the message space of $\mathcal{HE}$. 
Then, the matched payloads $\{T[j]\}$ are now additively shared between the two parties, with \PB{} holding $T[j] + r_j$ and \PA{} holding $-r_j$, while \PB{} is unaware of the matching positions $j \in {J}^1_B \cup {J}^2_B$.

\subsubsection{Payloads from Both Parties} 
Some ad conversion computations involve payloads from both the publisher and the advertiser.
Our protocol also generalizes easily to support payloads from both sides. 
Particularly, we make the following modifications over Figure~\ref{fig:proposed_ddh} to support payloads from both sides:
\begin{enumerate}
  \item \PA{} now generates the AHE key-pair and sends the public key to \PB{} in Step 4.
  \item \PA{} also encrypts its payloads using its public key and attaches the AHE ciphertexts in Step~4.
  \item \PA{} now needs to send the tags ${\sf rTag}^2_{B}$ to \PB{} in Step~15 in a shuffled order. 
        This enables \PB{} to obtain the matching positions on the second identifier ${J}^2_A$.
        By this point, \PB{} is already able to pick up \PA{}'s (encrypted) payloads in the intersection, \ie, ${J}^1_A \cup {J}^2_A$.
\end{enumerate}

\subsection{Concrete Constructions of $\F{\PRF}$}
We present two constructions of $\F{\PRF}$ using Hashed Diffie-Hellman 
(HashDH) and Yao's Garbled circuit~\cite{yao1986generate}, respectively.
Particularly, we leverage a functionality $\F{GC}$ 
which takes inputs $x$ and $y$ from
a circuit generator and a circuit evaluator, respectively, 
and returns the evaluation $BC(x, y)$ to the circuit evaluator for 
a Boolean circuit $BC: \bin^* \times \bin^* \mapsto \bin^*$ that is agreed upon by both the generator and the evaluator.

\figbox{Protocol $\Prot{\PRF}^{\sf DDH}$ from HashDH}{
\small

{\bf Initialization:} 

\PA{} samples $k_A, k'_A \usample \mathbb{Z}_q$.
\PB{} samples $k_B, k'_B  \usample \mathbb{Z}_q$.

{\bf Evaluation:~\cite{DBLP:conf/scn/JareckiL10} } $P_A$ inputs $X = \{x_i\}_{i=1}^{m}$.
\begin{enumerate}
  \item \PA{}: $a_i \leftarrow \hash(x_i)^{k_A}$ for $x_i \in X$.
  \item $\PA{} \rightarrow \PB{}: \{a_i \in \GG \}_i$.
  \item \PB{} outputs $Z = \{a_i^{k_B}\}$.
\end{enumerate}

{\bf Update:~\cite{DBLP:conf/uss/EverspaughCSJR15} } $P_B$ selects a subset $Z' \subset Z$.
\begin{enumerate}
  \setcounter{enumi}{3}
  \item $\PB{}: b_j \in \GG\leftarrow z_j^{k'_B/k_B}$ for $z_j \in Z'$.
  \item $\PA{} \leftarrow \PB{}: \{b_j\}_j$ in a shuffled order.
  \item $\PA{}: b'_j  = b_j^{k'_A/k_A}$.
  \item $\PA{} \rightarrow \PB:  \widetilde{Z} = \{b'_j\}_j$.
  \item \PB{} outputs $\widetilde{Z}$ by undoing the shuffle in Step 2. 
\end{enumerate}

}{
The $\PRF$ protocol from HashDH.
Both parties agree upon a group $\GG$ of a prime order $q$ that the DH assumption holds, and
a hash function $\hash: \bin^*\to \GG$ that maps bit strings to elements of $\GG$.
The corresponding PRF is defined as $F_k(x) = \hash(x)^k$.
\label{fig:DDH_OPRF}
}

Specifically,  in the HashDH-based construction,  we leverage a PRF function defined as $F_k(x) = \hash(x)^k$ where $\hash$ is a function 
that maps bit strings to elements of $\GG$.

\begin{theorem}
  If $\hash$ is modeled as a random oracle, and DDH assumption holds for the group $\GG$,
  the protocol $\Prot{\PRF}^{\sf DDH}$ in Figure~\ref{fig:DDH_OPRF} securely realizes the functionality $\F{\PRF}$ with the PRF defined as $F_k(x) = \hash(x)^k$.
\end{theorem}

We give a sketch proof here, and defer the formalized proof to Appendix.
\begin{proof}(Sketch.)
The proof follows the arguments in~\cite{DBLP:conf/scn/JareckiL10,DBLP:conf/uss/EverspaughCSJR15}.
w.l.o.g, we assume \PA{} acts as the Sender with input $X$ and \PB{} acts as the Receiver who
will invoke the {\bf Update} command.

\noindent{\bf Correctness.} 
  For the evaluation part, \PB{} obtains $a_i = \hash(x_i)^k_A$ which is then lifted to $a_i^{k_B}$ and forms $\hash(x_i)^{k_Ak_B\bmod q}$.
  For the update part, 
\[
  \left((\hash(x_j)^k)^{k'_B/k_B}\right)^{k'_A/k_A} = (\hash(x_j)^{k'_Bk_A})^{k'_A/k_A} = \hash(x_j)^{k'_Bk'_A},
\]
which is the PRF evaluation $F_{k'}(x_j)$ under the new key $k' = k'_A \cdot k'_B$.

\noindent{\bf Sender's (\PA{}) Privacy.}  
  For the evaluation part, we claim that 
  the views of the receiver (\ie, \PB{}) - \ie, $a_i = \hash(x_i)^{k_A}$ for $i = 1, \cdots, m$, where $\hash$ is modeled as a random oracle - 
  is indistinguishable from $r_1,\cdots, r_m$ with $r_i \usample \GG$.
  For the update part, 
  \PB{} knows $\{b_j\}$ but \PB{} can not 
  distinguish $\{(b_j)^{k'_A/k_A}\}_j$ from random elements from $\GG$ under the DDH assumption and random oracle model.

\noindent{\bf Receiver's (\PB{}) Privacy.} 
  For the evaluation part, \PB{}'s privacy is achieved directly from the CDH assumption.
  That is, \PA{} who knows $a_i$ can not derive \PB's secret $k_B$ after seeing $a_i^{k_B}$.
  For the update part, the views of the sender (\PA{}), \ie, $(\hash(x_j)^{k_A})^{k'_B}$, is indistinguishable from random group elements from $\GG$.
\end{proof}

\figbox{Protocol $\Prot{\PRF}^{\sf GC}$ from $\F{GC}$}{
\small

{\bf Initialization:} 

\PA{} samples $k_A, k'_A \usample \{0, 1\}^\lambda$.
\PB{} samples $k_B, k'_B  \usample \{0, 1\}^\lambda$.
Both parties aggree on a hash function $\hash': \bin^*\to \{0, 1\}^\lambda$.

{\bf Evaluation: } $P_A$ inputs $X = \{x_i\}_{i=1}^{m}$.
\begin{enumerate}
  \item \PA{}: $a_i \leftarrow \hash'(x_i)$ for $x_i \in X$.
  \item \PA{} and \PB{} invokes $\F{GC}$ to evaluate the standard AES encryptions ${\sf AES}_{k}(\cdot)$. 
    Particularly, \PA{} is the circuit generator with inputs $(k_A, \{a_i\})$, 
    and \PB{} is the circuit evaluator with inputs $k_B$.
    After the GC execution, \PB{} obtains $Z = \{ {\sf AES}_k(a_i)\}_{i=1}^{m}$ where the key is defined as $k = k_A\oplus k_B$.
\end{enumerate}

{\bf Update: } $P_B$ inputs $Z' \subset Z$.

\PA{} and \PB{} invokes $\F{GC}$ to evaluate the standard AES encryption.
Particularly, \PA{} is the circuit generator with inputs $k'_A$, 
and \PB{} is the circuit evaluator with inputs $(k'_B, Z')$.
After the GC execution, \PB{} obtains $\{{\sf AES}_{k'}(z_j)\}_{z_j \in Z'}$ with the key $k' = k'_A\oplus k'_B$.
}{
The $\PRF$ protocol from Garbled Circuit (GC).
The corresponding PRF is defined as $F_k(x) = {\sf AES}_k(\hash(x))$ where $\hash: \bin^*\to \{0, 1\}^\lambda$ that maps bit strings to an AES block.
\label{fig:GC_OPRF}
}

For the GC-based construction in Figure~\ref{fig:GC_OPRF}, we use the standard AES encryption ${\sf AES}_k(\hash(x))$ as a PRF function, where $\hash$ maps the input to an AES block.
The circuit generator/evaluator paradigm of GC naturally allows an oblivious evaluation of 
the PRF in a reversed manner by having the PRF receiver act as the circuit evaluator.
For the update part, we \textbf{do not} perform the key-switching 
from ${\sf AES}_k(x)$ to ${\sf AES}_{k'}(x)$, which might require a circuit 
twice the size of the evaluation part. 
Instead, we define the PRF after the update by double AES encryption 
$
F_{k'}(x) = {\sf AES}_{k'}({\sf AES}_k(\hash'(x))),
$
which also follows our interface by treating the inner evaluation ${\sf AES}_k(\hash'(x))$ as a hash from $\bin^* \mapsto \bin^\lambda$, given that the inner AES encryption acts as a pseudorandom permutation.

\begin{theorem}
  If $\hash'$ is modeled as a random oracle, 
  the protocol $\Prot{\PRF}^{\sf GC}$ in Figure~\ref{fig:GC_OPRF} securely 
  realizes the functionality $\F{\PRF}$ under the $\F{GC}$ hybrid. 
  The PRF is defined as $F_k(x) = {\sf AES}_k(\hash'(x))$.
\end{theorem}
The correctness and security simply reduce to the correctness and security of $\F{GC}$.

\subsubsection{Complexities}
Let $n = |X|$, $n' = |Z|$ and $C$ be the size of the AES encryption circuit. 
In total $\Prot{\PRF}^{\sf DDH}$ requires computing $2(n + n')$ exponentiations.
$\Prot{\PRF}^{\sf DDH}$ exchanges $O(n(|\GG| + 1))$ bits for the PRF evaluation, and $O(n'(|\GG| + 1))$ bits for the update step.
On the other hand, $\Prot{\PRF}^{\sf GC}$ exchanges $O(2nC\lambda + \lambda^2)$ bits for the PRF evaluation, and $O(n'C\lambda + \lambda^2)$ bits for the update step.

\section{DP Mechanism for Waterfall Matching}\label{sec:dp}

\subsection{DP Mechanism Design} \label{sec:single_id_dp}
We first consider a mechanism to protect the intersection-size for the single identifier setting, and then we extend it to the waterfall matching setting.
Randomness is introduced by padding dummy entries to both input sets $A$ and $B$ to create $\widetilde{A}$ and $\widetilde{B}$, respectively. 
This is done by randomly selecting and adding $\tau$ dummy entries from a common set $D$ with $2\tau$ rows, which is disjoint from both $A$ and $B$.
The specific steps are outlined in Algorithm~\ref{alg:DP-SI-SingleID}.

The following Theorem~\ref{th:single_id_dp} provides a simple way to enhance the privacy of any private matching protocol that leaks the intersection size.
For a given target privacy profile $(\epsilon, \delta)$, one can calculate the parameter $\tau$ from \eqref{eq:single_id_dp_profile}, which defines the size of the dummy set $D$.
\begin{theorem}\label{th:single_id_dp}
Let $\Pi$ be a semi-honest secure two-party protocol that takes private input sets $A$ and $B$ from \PA{} and \PB{}, respectively, and computes the intersection size, \ie $|A\cap B|$.
Also, let $D$ be a common set of dummies of $2\tau$ entries such that $D\cap A = \emptyset$ and $D\cap B = \emptyset$.
Then the composited protocol 
$\Pi'(A, B) := \Pi(\Mone(A, D), \Mone(B, D))$ is a semi-honest secure two-party protocol offering $(\epsilon, \delta)$-differential privacy,
where DP profile is given by
\begin{align}
\delta(\epsilon) = \frac{1}{\Comb{2\tau}{\tau}}
\left(1 +\sum_{z=\left\lceil\frac{\tau\cdot e^{\epsilon/2}-1}{{e^{\epsilon/2}+1}}\right\rceil}^{\tau-1}(\Comb{\tau}{z})^2 - e^{\epsilon}\cdot (\Comb{\tau}{z+1})^2\right).
\label{eq:single_id_dp_profile}
\end{align}
\end{theorem}

\begin{algorithm}[t]
\begin{small}
\SetKwInOut{Input}{Input}
\SetKwInOut{Output}{Output}
\Input{Set $S$ of $n$ entries.
       Set $D$ of $\tau' > \tau$ entries such that $D\cap S = \emptyset$.}
\Output{A new set $\tilde{S}$ of $\tilde{n} = n + \tau$ entries.}

\begin{algorithmic}[1]
\STATE $D' \subset D$ by randomly selecting a subset of $\tau$ entries.

\STATE Output $\tilde{S} = S + D'$ by appending $D'$ to $S$.
\end{algorithmic}
\end{small}
\caption{$\Mone$ Mechanism for single ID}
\label{alg:DP-SI-SingleID}
\end{algorithm}

To prove Theorem~\ref{th:single_id_dp}, we utilize the following.
We write $\mathcal{M}(A, B)$ to denote a variable of the intersection size between two (single-ID) sets.
Note that, up to this point, $\mathcal{M}(A, B)$ is independent of the specific private matching protocol.
\begin{lemma}
\label{lemma:set_cdn_pmf} 
The size of the intersection between two padded sets in Algorithm~\ref{alg:DP-SI-SingleID} can be represented as
\begin{equation}
\label{eq:dpca}
\mathcal{M}(A, B) = |\Mone(A, D) \cap \Mone(B, D)| = |A\cap B| + z, 
\end{equation}
where $z$ with axisymmetrially Probability Mass Function (PMF):
\begin{align}
\label{eq:z_dpca}
\Pr(z = z')= \begin{cases}
\frac{(\Comb{\tau}{z'})^2}{\Comb{2\tau}{\tau}} \quad &0 < z' \leq \tau, \\
0 \quad & \text{others}.
\end{cases} 
\end{align}
where $\Comb{n}{r}$ is the combinations out of a group of $n$.
\end{lemma}
In cases where we concentrate on one side, say $A$, the inclusion of $B$ is omitted from $\mathcal M(A, B)$ for simplicity, \ie, only $\mathcal{M}(A)$.

\begin{proof}(Theorem~\ref{th:single_id_dp})
We introduce the privacy loss random variable~\cite{dwork2016concentrated} for a tight DP analysis. 

\begin{definition} 
\label{def:pld}
The privacy loss random variable of $\mathcal M(A)$ over $\mathcal M(A^\prime)$ for any pair of neighboring inputs $A\sim A^\prime$ is defined as follows:

\noindent (i) If both $\Pr(\mathcal{M}(A)=o)\neq 0$ and $\Pr(\mathcal{M}(A^{\prime})=o)\neq 0$, then
\begin{equation}
\label{eq:PLD}
\bm \gamma_{AA^\prime} =
\ln\left(\frac{\Pr(\mathcal{M}(A)=o)}{\Pr(\mathcal{M}(A')=o)}\right).
\end{equation}
(ii) If  $\Pr(\mathcal{M}(A^{\prime})=o) = 0$ and $\Pr(\mathcal{M}(A^{\prime})=o) \neq 0$, then
$\bm \gamma_{AA^\prime} = \infty$.
\end{definition}
The quantity $\bm \gamma_{AA^\prime}$  
is the logarithmic ratio between the probability of observing outcome $o$ on input $A, B$ compared to input $A', B$, which is referred to as the privacy loss variable. From Definition~\ref{def:pld}, it is apparent that the privacy analysis involves both $\bm \gamma_{AA'}$ and $\bm \gamma_{A' A}$  by considering the order of $A$ and $A^\prime$.
According to Definition~\ref{def:dp}, $\mathcal M$ assures tightly $(\epsilon,\delta)$-DP with $\delta$ computed as follows:
$\delta(\epsilon)=\max \left(\delta_{AA^{\prime}}(\epsilon), \delta_{A^{\prime}A}(\epsilon)\right)$ 
where
\begin{align}\label{eq:6}
  &\delta_{AA^{\prime}}(\epsilon) = \\
  &\sum_{o \in O} \max \left(\Pr(\mathcal M(A)=o)-e^{\epsilon} \Pr(\mathcal M(A^\prime)=o), 0\right),\nonumber
\end{align}

\begin{align}\label{eq:7}
  & \delta_{A^{\prime}A}(\epsilon)  = \\
  & \sum_{o \in O} \max \left(\Pr(\mathcal M(A^\prime)=o)-e^{\epsilon} \Pr(\mathcal M(A)=o), 0\right).\nonumber
\end{align}

To streamline the analysis, we demonstrate the property that when employing a  DP mechanism with axisymmetrially distributed additive noise such as that proposed in Algorithm~\ref{alg:DP-SI-SingleID}, the order of neighboring data sets can be disregarded in the privacy analysis. 

\begin{property}\label{prop:symmetric}
\label{prop:symmetric}
If and only if the PMF of $\mathcal{M}(A)$ is axisymmetrically distributed, as exemplified by $\Pr(z)$ in \eqref{eq:z_dpca}, the following relationship holds for any non-negative $\epsilon$:
$\delta_{A^\prime A}(\epsilon)=\delta_{AA^\prime}(\epsilon)$.
\end{property}

Property~\ref{prop:symmetric}  simplifies the DP analysis by allowing the privacy accountant with either \eqref{eq:6} or \eqref{eq:7}. Therefore,  simplify $\bm \gamma_{AA'}$ as $\bm \gamma$. 
Substituting \eqref{eq:z_dpca} into \eqref{eq:PLD}, we have
\begin{equation}
\label{eq:PLD1}
    \Pr_{\bm \gamma}(\bm \gamma = \ln(\Comb{\tau}{z}/\Comb{\tau}{z+1})) = \frac{(\Comb{\tau}{z})^2}{\Comb{2\tau}{\tau}}, \quad \forall z\in[0, \dots, \tau].
\end{equation}
Meanwhile, \eqref{eq:6} can be reformulated as
\begin{equation}
\label{eq:dp_profile}
\begin{aligned}
\delta(\epsilon) & \geq \mathbb{E}_{\bm \gamma}[\max \{0,1-\exp (\epsilon-\gamma)\}] \\
& =\delta(+\infty)+\int_\epsilon^{\infty}(1-\exp (\epsilon-\gamma)) \Pr_{\bm \gamma}(\gamma) d \gamma.
\end{aligned}
\end{equation}
In this context, $\delta(+\infty)$ addresses failures resulting from the support discrepancy between $\mathcal{M}(A)$ and $\mathcal{M}(A^\prime)$, where a single occurrence leads to infinite leakage, \ie,
\begin{equation}
\label{eq:inf_delta}
\delta(+\infty)=\Pr(\mathcal{M}(A)=|A\cap B|).
\end{equation}   
By further substituting \eqref{eq:inf_delta} into the second term in the right-hand-side of \eqref{eq:dp_profile}, we acquire the intersection size's DP profile using Algorithm~\ref{alg:DP-SI-SingleID}, as outlined in Theorem~\ref{th:single_id_dp} below.
\end{proof}

\begin{algorithm}[t]
\begin{small}
\SetKwInOut{Input}{Input}
\SetKwInOut{Output}{Output}
  \Input{ 
  Sets $\{S_\ell \}_{\ell = 1}^{m}, \{ D_\ell \}_{\ell=1}^{m}$, and $\{ \hat{D}_\ell \}_{\ell=1}^{m}$.
  $S_\ell$ is a set of $n$ entries.  
  $D_\ell, \hat{D}_\ell$ are sets of $2\tau$ and $2\tau m$ entries, respectively.
  }
  \Output{New sets $\{\tilde{S}_\ell\}_{\ell=1}^{m}$ of $n + m\tau$ entries each.}
\begin{algorithmic}[1]
  \FOR{$1\le k \le m$}
  \STATE Initialize $\tilde{S}_k = S_k$.
  \FOR{$1\le \ell < k$}
  \STATE Update the set $\tilde{S}_k \leftarrow \Mone(\tilde{S}_k, \hat{D}_\ell)$
  \ENDFOR{}
  \STATE Update the set $\tilde{S}_k \leftarrow \Mone(\tilde{S}_k, D_k)$
  \FOR{$k < \ell \le m$}
  \STATE Update the set $\tilde{S}_k \leftarrow \Mone(\tilde{S}_k, \hat{D}_\ell)$
  \ENDFOR{}
  \ENDFOR{}
\end{algorithmic}
\end{small}
\caption{$\Mwf$ Mechanism for multiple IDs}
\label{alg:DP-SI-MultiID}
\end{algorithm}

\subsection{DP Mechanism for the Multi-ID Matching} \label{sec:multi_id_dp}
DP protection in multi-ID scenarios can be accomplished by applying the single-ID DP mechanism to each ID column. 
It is well-known that the composition of DP mechanisms compromises the privacy of DP protection.
Suppose $\Mone$ (\ie, Algorithm~\ref{alg:DP-SI-SingleID}) 
achieves $(\epsilon_\ell,\delta_\ell)$-DP for its intersection size for the $\ell$-th ID column.
The basic composition reveals that DP protection, considering the observation of each intersection size of each ID column from the two parties, is 
$
\left(\epsilon, \delta\right)$-DP, where 
$\epsilon=\sum_{\ell=1}^m\epsilon_\ell$ and $\delta=\sum_{\ell=1}^m \delta_\ell$.
Consequently, the protective capacity of DP diminishes linearly with an increase in the number of IDs.

The linear rate of diminishing DP is attributed to the fact that different IDs matched for each column could correspond to the same record. 
If there is a way to ensure that the cross-ID dummies cover disjoint subsets of the data records, the sequential composition is no longer a tight DP analysis. 
Instead, the combined privacy loss is the maximum of $\epsilon_\ell$ for all $\ell\in [1, m]$ from the parallel composition theorem.
To ensure that cross-ID dummies cover disjoint subsets of the dummy records, we propose  Algorithm~\ref{alg:DP-SI-MultiID} to construct dummy variables. The sufficient condition for parallel composition is that the dummies padded for different ID columns come from disjoint spaces and cannot match.

\begin{property} (DP profile for Waterfall Matching.)
In Algorithm~\ref{alg:DP-SI-MultiID}, if the mechanism $\Mone$ provides achieving  $(\epsilon, \delta)$-DP for each ID column,
then the $\Mwf$ mechanism also achieves $(\epsilon, \delta)$-DP for the waterfall matching.
\end{property}

\subsection{DP-enhanced Waterfall Matching Protocol}
The DP-enhanced waterfall matching $\Prot{dp{-}wmf}$ operates as follows.

\begin{enumerate}
  \item \PA{} and \PB{} compute the dummies size $\tau$ according to \eqref{eq:single_id_dp_profile} under the DP profile $(\epsilon, \delta)$.
  \item \PA{} and \PB{} agree on dummy sets ${\sf Dmy} = (D_1, \hat{D}_1, D_2, \hat{D}_2)$.
         Particularly, 
         $|D_1| = |D_2| = 2\tau$ and $|\hat{D}_1| = |\hat{D}_2| = 4\tau$, and there are disjoint to each other, \ie,
    $D_1 \cap D_{2} = \emptyset$, 
    $\hat{D}_1 \cap \hat{D}_{2} = \emptyset$,
    $D_1 \cap \hat{D}_{2} = \emptyset$, and
    $\hat{D}_1 \cap {D}_{2} = \emptyset$.

  \item \PA{} samples $\{\widetilde{\ID}_A^b\}_b \leftarrow \Mwf(\{\ID_A^b\}_b, \{D_b\}_b, \{\hat{D}_b\}_b)$ according to Algorithm~\ref{alg:DP-SI-MultiID}.
  \item \PB{} samples $\{\widetilde{\ID}_B^b\}_b \leftarrow \Mwf(\{\ID_B^b\}_b, \{D_b\}_b, \{\hat{D}_b\}_b)$ according to Algorithm~\ref{alg:DP-SI-MultiID}.
  \item \PB{} appends $\widetilde{T} \leftarrow T\lVert 0$ with zeros to $|\widetilde{T}| = |\widetilde{\ID}_B^1|$.
  \item \PA{} and \PB{} jointly run $\Prot{wmf}$ on the modified inputs, \ie,
  $(\widetilde{\ID}_A^1, \widetilde{\ID}_A^2)$ and $(\widetilde{\ID}_B^1, \widetilde{\ID}_B^2, \widetilde{T})$.
\end{enumerate}


This DP-enhanced protocol $\Prot{dp{-}wmf}$ is clearly a secure two-party protocol as long as the base protocol $\Prot{wmf}$ is secure. 
Note that \PA{} and \PB{}'s views in $\Prot{dp{-}wmf}$ differ from their views in the base protocol $\Prot{wmf}$ only with respect to the intersection size. 
Thus, the $\Prot{dp{-}wmf}$ protocol ensures $(\epsilon, \delta)$-differential privacy (Definition~\ref{def:2pc_dp}) for the intersection sizes by a similar composition argument in Theorem~\ref{th:single_id_dp}.

\noindent{\bf Discussions.} 
$\PA$ and $\PB$ can generate a common set of randomness by sharing a random seed.
Moreover, to make sure the disjoint requirements (\ie, $D\cap A = \emptyset$ and $D\cap B = \emptyset$), they can 
leverage the specific type of the identifier.
For instance, suppose the target identifier is the "phone numbers", then $\PA$ and $\PB$ can generate the dummies from the alphabet set,
or they can generate the dummies with longer digits than a valid phone number.

\subsection{A Tighter Privacy Profile for Multiple Executions of Waterfall Matching}
There are real needs to perform matching over the same ID sets multiple times. For instance:
\begin{itemize}
  \item The ad publisher and the advertiser may perform the matching twice to aggregate ad conversion data across two ad campaigns launched in close proximity.
  \item The advertiser may want to change the payload and perform one more matching after having seen the aggregated value from the first matching.
\end{itemize}

From the privacy perspective, the notion of combining multiple DP intersection results while preserving the overall privacy guarantee is referred to as the sequential DP composition problem. 
One widely adopted analysis is the sequential composition, revealing a linear increasing rate of $\epsilon$. Given the inversely proportional relationship between the dummy size $\tau$ and $\epsilon$, it is natural to explore a tight DP composition that decreases $\tau$ the number of dummies
under the same privacy budget $\epsilon$.
We adapt the composition results shown in 
\cite{koskela2020tight} to our multiple execution cases.
Given $\Pr_{\bm \gamma}$ in \eqref{eq:PLD1}, 
after $k$ times of execution of our waterfall matching protocol, the composition is tightly $(\epsilon, \delta)$-DP for $\delta^m(\epsilon)$ given by
\begin{equation}\label{eq:delta^k}
\begin{aligned}
\delta^k(\epsilon) & =1-\left(1-\delta(+\infty)\right)^k + 
\int_{\epsilon}^{\infty}\left(1-\mathrm{e}^{\epsilon-\gamma}\right)\left(\Pr_{\bm \gamma} \ast^k \Pr_{\bm \gamma}\right)(\gamma) \mathrm{d} \gamma,
\end{aligned}
\end{equation}
where $\delta^k(\epsilon)$ denotes the post-composition privacy profile, and $\ast^k$ denotes the $k$ times convolution.
According to \eqref{eq:inf_delta}, it is straightforward that,
\begin{equation}\label{eq: k-fold}
\begin{aligned}
    \delta^k(\epsilon) &= 1- \left(1-\frac{1}{\Comb{2\tau}{\tau}}\right)^k
    + \sum_{\gamma:\epsilon<\gamma<\infty}(1-e^{\epsilon-\gamma})\Pr_{\bm \gamma}(\gamma)^{\circledast k},\\
\end{aligned}
\end{equation}
where $(\cdot)^{\circledast k} $ denotes the $k$-fold convolution. 
Note that directly solving $\Pr_{\bm \gamma}(\gamma)^{\circledast k} $ exhibits exponentially computation complexity w.r.t. $k$. 
Following similar steps as presented in \cite{koskela2020tight}, we present the two-party DP composition accounting algorithm based on FFT in Algorithm~\ref{alg:bisection}, subsequently, we outline steps to numerically derive $\tau$ that guarantees $(\epsilon,\delta)$-DP after $k$-fold composition in Algorithm~\ref{alg:bisection}.












\begin{algorithm}[t]
\SetKwInOut{Input}{Input}
\SetKwInOut{Output}{Output}
\Input{
Total DP budget $(\epsilon, \delta)$;
Total numbers of executions: $k$;
}

\Output{The number of dummies added for DP: $\tau$.}

\begin{footnotesize}
\begin{algorithmic}[2] 
\PROCEDURE{Estimate-$\delta$}{$\tau, k, \epsilon, n_x$}
\STATE FFT window: $W = (4k - 2)\ln{\tau}$.
\STATE FFT resolution: $\Delta_x = 2W/n_x$.

\STATE Initialize PLD vector as an all zeros vector $P_{\Gamma}\in \mathbb R^{n_x}$.

\FOR{$o\in [0, \dots, \tau-1]$:}

\STATE Update:  $
    P_{\Gamma}\left[\lceil (W + 2\ln{\Comb{\tau}{o}} - 2\ln{\Comb{\tau}{o+1}})/\Delta_x\rceil\right]={(\Comb{\tau}{o})^2}/{\Comb{2\tau}{\tau}} \nonumber
$

\ENDFOR

\STATE Compute the convolutions via FFT: $
 b = D\mathcal{F}^{-1}\left(\mathcal{F}(D\cdot P_{\Gamma})^{k}\right)
,$ where
$D = \begin{bmatrix}
0 & I_{n_x/2} \\
I_{n_x/2} & 0 
\end{bmatrix}$.\\

\STATE Compute the starting point:
$\gamma_{\epsilon}=\left\lceil\frac{\epsilon+W}{\Delta_x}\right\rceil. \nonumber
$

\STATE Return $\delta' = 1- \left(1-\frac{1}{\Comb{2\tau}{\tau}}\right)^k + \sum_{\gamma=\gamma_{\epsilon}}^{n_x}(1-e^{\epsilon+W-\gamma\Delta_x})b[\gamma]$.
\ENDPROCEDURE
\end{algorithmic}
\end{footnotesize}

\begin{algorithmic}[1]
\STATE Set $\tau_{\rm lo} = 0$ 
\STATE Increase $\tau_{\rm hi} = 0, 1, \cdots, $ until $\delta \le \textsc{Estimate-$\delta$}(\tau_{\rm hi}, k, \epsilon, n_x)$.
\WHILE{$\tau_{\rm lo} < \tau_{\rm hi}$}
\STATE Set $\tau' = \ceil{(\tau_{\rm lo} + \tau_{\rm hi})/2}$.
\STATE $\delta' \leftarrow \textsc{Estimate-$\delta$}(\tau', m, \epsilon, W, n_x)$.
\STATE Set $\tau_{\rm lo} \leftarrow \tau'$ if $\delta' \ge \delta$. Otherwise set $\tau_{\rm hi} \leftarrow \tau'$ if $\delta' < \delta$.
\ENDWHILE
\STATE Output $\tau_{\rm hi}$.
\end{algorithmic}

\caption{Find the minimal $\tau$ for multiple executions.\label{alg:bisection}}
\end{algorithm}

FFT-based composition for privacy loss involves representing the distribution as a sparse vector, marking specific locations indicating data leakage and their probabilities. This vector is transformed into the frequency domain using FFT, followed by multiplication operations to handle convolution. Matrix expressions simplify these complex operations into linear transformations and power computations. The processed sequence is then reverted to the time domain using the inverse FFT transformation.

FFT significantly improves the computation cost  from $\mathcal{O}(n^2)$ to $\mathcal{O}(n\log(n))$, where $n$ denotes the length of $\Pr_{\bm \gamma}(\gamma)^{\circledast k}  $. However, FFT requires the input sequence to be periodic. Further, the input sequence should be cast to a periodical window with size, say $W$, no less than the length of FFT result to avoid information loss due to spectrum overlap. 
To this end, we proved the following property with details in Appendix, and it is used to set the parameter $W$ required in our waterfall matching protocol.


\begin{property}(Minimal FFT window)
\label{prop: window}
The minimal window size for FFT to guarantee no spectrum overlap is
$W \ge (4K-2)\ln\tau$.
\end{property}

Further, casting continuously-valued leakages to discrete vector locations incurs information loss. The discretization is handled by the resolution factor $n_x$ that represents the total number of entries on the input sequence vector. Intuitively, larger $n_x$ leads to high input resolution, and thereby reduces information loss. However, large $n_x$ inevitably increases the computation complexity. 
By examining line 6 of the \textsc{Estimate-$\delta$} procedure in Algorithm~\ref{alg:bisection}, this line projects the continuously valued leakage onto the closest discretized location to its right. In essence, this means that the leakage is consistently amplified or increased, thereby ensuring that the calculated $\delta'$ is always greater than or equal to the true $\delta$ that accurately reflects the exact leakage. 

In summary, for a given target $(\epsilon, \delta)$, one can calculate the parameter $\tau$ using Algorithm~\ref{alg:bisection}.




\section{Experiments}\label{sec:experiments}

\begin{table}[t]
\caption{
The minimum dummies size $\tau$ is a function $\tau = \tau(n, m, \epsilon, \delta, k)$ of 
$n$ the input size $n$, $m$ the number of IDs,
$(\epsilon, \delta)$ the DP profile, and $k$ the maximum sequential matching.
The below are fixed by $m = 3$.
\label{tbl:dummy_size}
}
\centering
\begin{tabular}{ccccc|ccccc}\toprule
\multicolumn{5}{c}{$k = 1, \epsilon = 1, \delta = 1/(10n)$} & \multicolumn{5}{c}{$k = 6, \epsilon = 1, \delta = 1/(10n)$} \\
$n$ & $10^4$ & $10^5$ & $10^6$ & $10^7$ & $n$ & $10^4$ & $10^5$ & $10^6$ & $10^7$ \\
$\tau$ & 114 & 141 & 170 & 201 & $\tau$ &  285 & 353 & 425 & 503 \\
\midrule
\multicolumn{5}{c}{$k = 1, \epsilon = 2, \delta = 1/(10n)$} & \multicolumn{5}{c}{$k = 6, \epsilon = 2, \delta = 1/(10n)$} \\
$n$  & $10^4$ & $10^5$ & $10^6$ & $10^7$ & $n$  & $10^4$ & $10^5$ & $10^6$ & $10^7$ \\
$\tau$& 77 & 96 & 116 & 137 & $\tau$ & 191& 239 & 289 & 340 \\
\bottomrule
\end{tabular}
\end{table}


\noindent{\bf Testbed.} 
All computational cost measurements for our protocols are in terms of total wall-clock runtime for both parties, running 
on cloud instances with an 
Intel(R) Xeon(R) Platinum 8336C CPU (2.30GHz) and 64GB of RAM.
Communication cost includes the inbound and out-bound traffice of the two parties.
We consider three network conditions, including LAN (10Gbps, 0.2ms ping), MAN (1Gbps, 2ms ping), and WAN (100Mbps, 20ms ping).

\noindent{\bf Synthetic Data.}
We consider balanced databases $|X| = |Y| = n$ for $10^3 \le n \le 10^7$ with 
the number of IDs $m = 2, 3$.
The total intersection size is set as $|X\cap Y| = 0.02n$.
The payloads are at most 32 bits long. 

\noindent {\bf Implementations of $\PRF$.}
For the DDH-based construction, we use OpenSSL's implementation ``prime256v1'', a NIST elliptic curve 
with 256-bit group elements, as the group $\GG$. 
For the random oracle, we use SHA-256 and apply the ``try-and-increment'' method to map bit-string to group elements~\cite{DBLP:conf/crypto/BonehF01}.
We mention a common trick for the DDH-based construction.
For the update command, we can truncate the updated group element (\ie, $b'_j$ in Step~4 in Figure~\ref{fig:DDH_OPRF}) 
to save some communication. For instance, we only send the least significant $\gamma$ bits of $b'_j$.
This is suffice for the private matching function given that the 
the least significant bits of a random element of $\GG$ are indistinguishable from a uniform bit-string~\cite{DBLP:conf/eurocrypt/ChevalierFPZ09}.
Particularly, we set $\gamma = 96$.
For the GC-based construction, we use the EMP-toolkit~\cite{emp-toolkit} and set the AES block size $\lambda = 128$.
For the random oracle, we use the BLAKE2b hash function to produce 128-bit digests.
For the AHE, we use the SEAL library~\cite{sealcrypto} for the BFV scheme with the AVX512 acceleration~\cite{IntelHEXL}.
Specifically, we use BFV with an 118-bit ciphertext modulus and a 64-bit plaintext modulus, with 8192 coefficients per ciphertext.
Each ciphertext is 236 KB\footnote{We can leverage the symmetric version of BFV to reduce the size to 118KB.} and can hold up to 8192 payloads of 64-bit each.

\subsection{Enhancing Privacy via DP}

\begin{figure}[t]
\centering
\subfigure[
Revealing the intersection sizes in the multi-ID setting are more susceptible to membership inference attacks. 
  \label{fig:single_vs_mult_id_with_DP}
  ]{
   \includegraphics[width=0.45\linewidth]{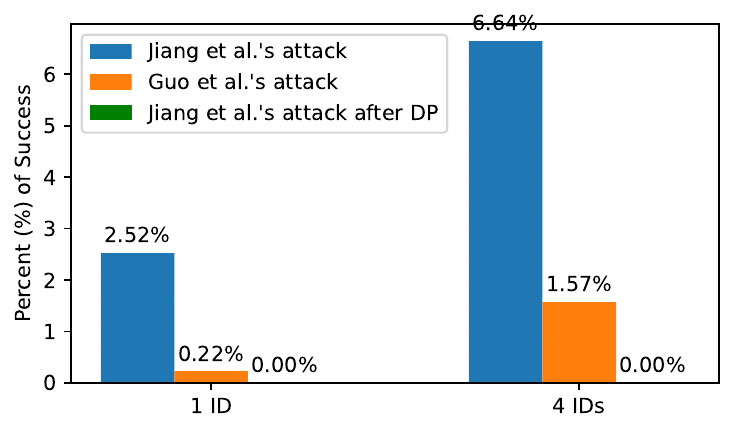}
  }
  ~
  \subfigure[
  The parallel composition introduces less dummies
  than the sequential composition under the profile ($\epsilon = 1.0, \delta = 10^{-6}$).
  \label{fig:linear_parallel}
  ]{
  \includegraphics[width=0.45\linewidth]{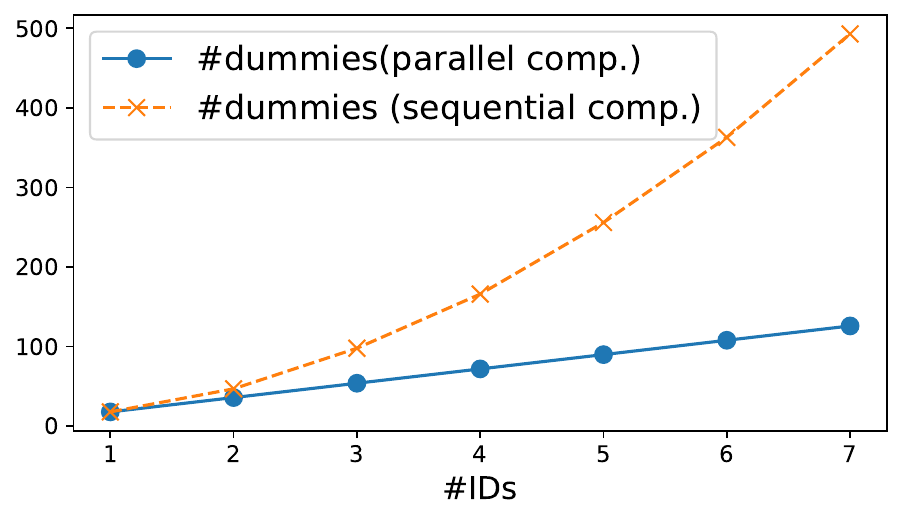}
  }
  \caption{(a) The robustness of the DP-enhanced protocol against the two membership inference attacks. 
           (b) The parallel composition introduces less dummies.
           }
\end{figure}

Table~\ref{tbl:dummy_size} shows the corresponding dummy sizes under different input settings according to Algorithm~\ref{alg:bisection}.
Briefly, the dummy size is significantly smaller than the input size by virtue of the parallel decomposition and FFT technique applied.
In Figure~\ref{fig:single_vs_mult_id_with_DP}, we demonstrate the robustness of the DP-enhanced protocol against two membership inference attacks. 
Specifically, it illustrates that the membership leakage after $k=20$ executions of a size-revealing private matching protocol can be up to 6\% using recent attacks~\cite{Guo2022birds, jiang2022comprehensive}. However, after applying our DP protection for the intersection size, we do not observe any successful inferences using these attacks.
Figure~\ref{fig:linear_parallel} depicts the effectiveness of the parallel decomposition for introducing less number of dummies.
This also translates to the FFT-based composition technique used for multiple executions of the waterfall matching.


\subsection{Micro-benchmarks}

Table~\ref{tbl:mirco_prf} depicts the micro-benchmarks for the two constructions of $\PRF$. 
Specifically, the DDH-based construction is more computation intensive but has significantly less communication overhead. 
On the other hand, the GC-based construction introduces a significantly larger communication overhead.

\noindent{\textbf{A Note on PSI Based on Symmetric Primitives}}
It is often perceived that the DDH-based approach may be less efficient than conventional symmetric-primitive-based PSI protocols. However, existing experimental results challenge this misconception. For instance, Peter et al. reported a circuit-based PSI implementation (single identifier) requiring 8 CPU cores and 103 seconds (277 MB memory) for a dataset of  $n=10^6$ records~\cite[Table 4]{rindal2021vole}. 
This demonstrates  $2\times$ higher CPU workloads and  $3.7\times$ greater communication overhead compared to our DDH-based method.

\begin{table}[t]
\caption{Benchmarks of $\PRF$\label{tbl:mirco_prf}. Single threaded. The time was measured under the LAN setting.}
  \centering
  \begin{footnotesize}
  \begin{tabular}{c rr|rr} \toprule
    \multicolumn{5}{c}{\underline{Send Command}}\\
     & \multicolumn{2}{c}{\underline{\hspace{5pt}DDH\hspace{5pt}}} 
     &  \multicolumn{2}{c}{\underline{\hspace{5pt}GC\hspace{5pt}}}\\
    {Input $n$} & Time & Commu. & Time & Comm.\\
       \midrule
       $10^3$ & 0.19s&  32.23KB   & 0.52s & 209.74MB \\
       $10^4$ & 1.88s&  322.26KB  & 5.21s & 2.05GB  \\
       $10^5$ & 18.56s&  3.15MB   & 51.23s& 20.46GB \\
       $10^6$ & 185.61s&  31.52MB &512.49s& 204.56GB \\
       \midrule
       \midrule
    \multicolumn{5}{c}{\underline{Update Command}}\\
        & \multicolumn{2}{c}{\underline{\hspace{5pt}DDH\hspace{5pt}}} 
        &  \multicolumn{2}{c}{\underline{\hspace{5pt}GC\hspace{5pt}}}\\
    {Input $n$} & Time & Commu. & Time & Comm. \\
       \midrule
       $10^3$ & 0.16s&  43.95KB   & 0.52s & 209.74MB\\
       $10^4$ & 1.62s&  439.45KB  & 5.21s & 2.05GB  \\ 
       $10^5$ & 16.48s&  4.29MB   & 51.23s& 20.46GB \\
       $10^6$ & 163.06s&  41.92MB & 512.49s& 204.56GB  \\
      \bottomrule
      \multicolumn{5}{l}{$^\dagger$ 1GB = $2^10$ MB = $2^{20}$ KB = $2^{30}$ bytes}
  \end{tabular}
  \end{footnotesize}
\end{table}

\begin{table}[t]
\caption{Benchmarks the proposed waterfall matching protocol using $\Prot{\PRF}^{\sf DDH}$.
8 threads were used for each party. \label{tbl:bench_waterfall}}
\centering

\begin{tabular}{crrrr}\toprule
   \multicolumn{5}{c}{\underline{Matching for Sum} (\#IDs $m = 2$)} \\
  \multicolumn{1}{c|}{$n$} & Commu. & LAN & MAN & WAN  \\ \midrule
  $10^4$ & 2.91MB& 1.89s& 2.03s& 3.51s \\
  $10^5$ & 26.61MB & 6.93s& 7.22s& 10.54s \\
  $10^6$ & 263.1MB & 68.01s & 69.86s & 88.18s \\
  $10^7$ & 2631MB & 697s& 719s& 911s \\
\midrule
  & \multicolumn{4}{c}{\underline{Matching for Share} (\#IDs $m = 3$)} \\
  \multicolumn{1}{c|}{$n$} & Comm. & LAN & MAN & WAN  \\ \midrule
  $10^4$  &4.14MB& 1.42s & 1.48s& 3.58s \\
  $10^5$  &41.36MB& 11.29s& 11.41s& 15.44s \\
  $10^6$  &410.64MB & 106.68s & 109.71s & 139.35s \\
  $10^7$  &4102MB& 1095s& 1121s& 1430s \\
  \midrule
  \midrule
  
   \multicolumn{5}{c}{\underline{Matching for Share} (\#IDs $m = 2$)} \\
  \multicolumn{1}{c|}{$n$} & Comm. & LAN & MAN & WAN  \\ \midrule
   $10^4$        & 3.16MB& 0.95s& 1.18s& 2.63s \\
   $10^5$        & 26.85MB&  6.93s&  7.73s& 10.98s \\
   $10^6$        & 237.19MB& 72.46s& 74.17s& 94.09s \\
   $10^7$        & 2637MB& 732s& 742s& 940s\\
\midrule
  & \multicolumn{4}{c}{\underline{Matching for Share} (\#IDs $m = 3$) } \\
  \multicolumn{1}{c|}{$n$} & Comm. & LAN & MAN & WAN  \\ \midrule
   $10^4$        & 4.64MB& 2.36s& 2.16s& 3.76s\\
   $10^5$        & 41.59MB& 11.99s& 12.34s& 16.69s\\
   $10^6$        & 411.36MB& 110.02s& 114.28s& 145.79s\\
   $10^7$        & 4109MB& 1125s& 1153s& 1456s\\
  \bottomrule
\end{tabular}
\end{table}


\subsection{Benchmark the Private Waterfall Matching}

Table~\ref{tbl:bench_waterfall} shows the performance of our waterfall matching protocol under different network conditions.
From the results, we see that the Matching-for-Sum and the Matching-for-Share variant do not change much.
That is because the $\PRF$ part (\ie, DDH-based $\PRF$) dominates the computation and communication.
The AHE part only takes a very small portion of the whole protocol for two reasons.
First, the BFV is very efficient for encrypting a long vector, \ie, million records per second.
Second, the intersection size is significantly smaller than the input size, leading a relatively smaller  workload for the AHE, \eg, addition and decryption.
Overall, our waterfall matching protocol is practical to handle large-scale input sets.

\subsection{Comparison with Existing Solutions}
Figure~\ref{fig:benchmark_other_method} shows a comparison with two existing approaches. 
To the best of our knowledge, none of the existing approaches fully realize waterfall matching without cross-ID leakages. 
This is because most of them are designed for the single ID setting.

\subsubsection{Private-Join-then-Sum (PJS)~\cite{google-psi}} 
Ion \etal present a matching-for-sum protocol using hashed DDH and the Paillier AHE scheme. 
Their approach works for a single ID and reveals the exact intersection size. 
We re-ran their implementation\footnote{\url{https://github.com/google/private-join-and-compute}} in our environment. 
Due to the superiority of the BFV scheme over the Paillier scheme, our approach is about $20\times$ faster, with approximately one-tenth the communication overhead compared to their approach.

\subsubsection{P$\text{S}^3$I~\cite{PS3I}, a matching-for-share protocol with payloads from both sides}
We also perform the benchmark under the single-ID setting with two payload columns, and we compare the performance 
using their implementation\footnote{\url{https://github.com/facebookresearch/Private-ID}}.
In general, our protocol is about 2 order of magnitude faster than  P$\text{S}^3$I with 85\% less communication. 
This improvement is a result of multiple factors.  P$\text{S}^3$I~\cite{PS3I} utilizes the Paillier scheme which is less efficient than the BFV scheme we used.
Also, P$\text{S}^3$I~\cite{PS3I} transfer all encrypted payloads while our approach only transfer the encrypted payload within the intersection.


\section*{Conclusions}
In this paper, we propose a practical and privacy-preserving solution to the multidimensional matching problem. By integrating a novel two-party distributed differential privacy mechanism with a specially tailored matching protocol, our approach effectively mitigates membership inference leakage. Our design satisfies the practical requirements for handling multiple identifiers while minimizing the overhead of additional dummy elements needed for privacy protection. This combination of robust privacy guarantees and computational efficiency represents a significant advancement in privacy-preserving data matching, offering both theoretical rigour and practical applicability for real-world deployment.

\begin{figure}
    \centering
    \includegraphics[width=0.90\linewidth]{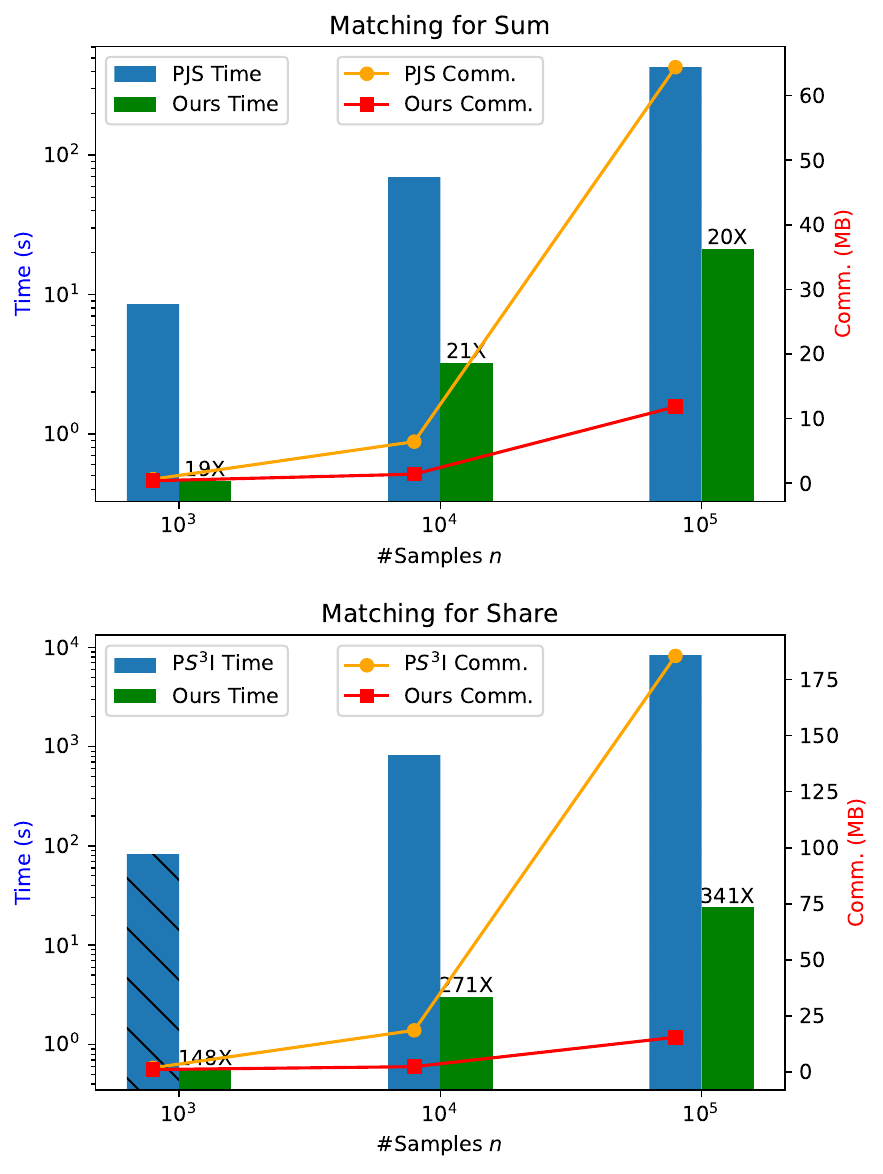}
    \caption{Comparison with the PJS~\cite{google-psi} protocol and the P$\text{S}^3$I~\cite{PS3I} protocol.
For our DDH-based protocol, we configure the DP profile with the parameters $(k = 6, m = 3, \epsilon = 2, \delta = 1/(10n))$.
The experiments were conducted under the LAN network setting, utilizing one single thread per party. \label{fig:benchmark_other_method}
    }
\end{figure}

\bibliographystyle{IEEEtran}
\bibliography{Jian}

\appendix

\section{Appendix} \label{sec:appendix_dp_proof}
\subsection{Proofs}

\subsubsection{Proof of Lemma~\ref{lemma:set_cdn_pmf}} \label{app:pmf}

\begin{proof}
We denote $D'_A$ and $D'_B$ the sampled set in Algorithm~\ref{alg:DP-SI-SingleID} by $\PA$ and $\PB$, respectively.
According to Algorithm~\ref{alg:DP-SI-SingleID}, it is established that both  $D'_A$ and $D'_B$ are independent of both $A$ and $B$. 
Consequently, it becomes evident that $|\widetilde{A}\cap \widetilde{B}| =|A\cap B|+ z$ where $z = |D'_A \cap D'_B|$.
To calculate $\Pr(z = z')$, we need to determine the number of outcomes where the intersection size of $D'_A \cap D'_B$ is $z$ and 
then divide by the total number of possible outcomes. 
There are $\Comb{2\tau}{\tau}$ possible outcomes for both $D'_A$ and $D'_B$ since there are random subsets of a dummy set with $2\tau$ elements.
Therefore, the total number of outcomes is $\Comb{2\tau}{\tau} \cdot \Comb{2\tau}{\tau}$.

The number of outcomes for $D'_B$ alone is $\Comb{2\tau}{\tau}$. 
For any $|D'_A \cap D'_B| = z$, it implies that there are $z$ elements in $D'_B$ belonging to $D'_A$, with $\Comb{\tau}{z}$ possible outcomes. 
Simultaneously, the remaining $\tau - z$ elements in $D'_B$ must be randomly selected from $D \setminus D'_A$, offering $\Comb{\tau}{\tau - z}$ possibilities. 
Consequently, we can compute the PMF of $z = |D'_A \cap D'_B|$ as follows:
$$
\label{eq:PMF}
\Pr(z= z') =  \frac{\Comb{2\tau}{\tau} \cdot \Comb{\tau}{z'} \cdot \Comb{\tau}{\tau - z'} }{\Comb{2\tau}{\tau} \cdot\Comb{2\tau}{\tau} }= \frac{(\Comb{\tau}{z'})^2}{\Comb{2\tau}{\tau}}.
$$
\end{proof}

\subsubsection{Proof of Property~\ref{prop:symmetric}} \label{app:one-delta}
\begin{proof}
Denote {$\{0, \dots, \tau\}$} as the support of $z$, the support of $\mathcal{M}(A)$ becomes: $\{|A\cap B|,...,|A\cap B|+\tau\}$.
Given the asymmetrically distributed property, we have $\Pr(z = \tau-o) = \Pr(z = o)$, $\forall{o\ in[0, 1, \dots, \tau]}$. 
Given the sensitivity of the operation of  $|A\cap B|$ is $1$, 
The worst-case neighboring dataset $\mathcal{M}(A^\prime)$ has the support of $\{|A\cap B|+1, \dots, |A\cap B| + \tau + 1\}$, and the mapping relationships of the two neighboring datasets become $\Pr(\mathcal{M}(A)=o) = \Pr(\mathcal{M}(A^\prime) = o+1)$ .
Define ${O}$ as the support of the output of the mechanism, ${O}^+$ as the subset of ${O}$ corresponds to the criterion:
\[
O^+ = \{o: \Pr(\mathcal M(A)=o)-e^{\varepsilon} \Pr(\mathcal M(A^\prime)=o)\ge 0 \}. 
\]
Then, according to the definition of DP, the failure probability $\delta$ can be expressed as:
\begin{equation*}\label{eq:delta_expression}
\begin{aligned}
    \delta=&\sum_{o \in O} \max \left(\Pr(\mathcal M(A)=o)-e^{\varepsilon} \Pr(\mathcal M(A^\prime)=o), 0\right)\\
    =&\sum_{o\in{O}^+}\left(\Pr(\mathcal M(A)=o)-e^{\varepsilon} \Pr(\mathcal M(A^\prime)=o\right)\\
    =&\sum_{o\in{O}^+}\left(1-e^{\varepsilon}\frac{\Pr(\mathcal M(A^\prime)=o)}{\Pr(\mathcal M(A)=o)}\right)\Pr(\mathcal M(A)=o),
    \end{aligned}
\end{equation*}
This shows how \eqref{eq:dp_profile} is derived.
For the first term in \eqref{eq:dp_profile}
\begin{equation}
\begin{aligned}
\delta_{AA'}(\infty)=&\sum_{\substack{\{o|\Pr(\mathcal{M}(A)=o)>0,\\ \Pr(\mathcal{M}(A^\prime)=o)=0\}}}\Pr(\mathcal{M}(A)=o)\\
=&\sum_{o=|A\cap B|}^{|A\cap B|}\Pr(\mathcal{M}(A)=o) \\
=&\Pr(\mathcal{M}(A)=|A\cap B|)
\end{aligned}
\end{equation}
For the other direction,
\begin{equation}
\begin{aligned}
\delta_{A'A}(\infty)=&\sum_{\substack{\{o|\Pr(\mathcal{M}(A^\prime)=o)>0,\\ \Pr(\mathcal{M}(A)=o)=0\}}}\Pr(\mathcal{M}(A^\prime)=o)\\
=&\sum_{o=|A\cap B|+\tau}^{|A\cap B|+\tau}\Pr(\mathcal{M}(A)=o)\\
=&\Pr(\mathcal{M}(A)=|A\cap B| + \tau)\\
=&\Pr(\mathcal{M}(A)=|A\cap B|)\\
\end{aligned}
\end{equation}
which means $\delta_{AA'}(\infty) = \delta_{A'A}(\infty)$.

For the second part of \eqref{eq:dp_profile}, first, let $\bm \gamma_{AA'}(o)$ and $\bm \gamma_{A'A}(o)$ be:
\begin{equation*}
\begin{aligned}
&\bm \gamma_{AA'}(o) = \ln\left(\frac{\Pr(\mathcal{M}(A)=o)}{\Pr(\mathcal{M}(A^\prime)=o)}\right),\\
&\bm \gamma_{A'A}(o) = \ln\left(\frac{\Pr(\mathcal{M}(A^\prime)=o)}{\Pr(\mathcal{M}(A)=o)}\right)
\end{aligned}
\end{equation*}
Then, $\bm \gamma_{AA'}(o)=-\bm \gamma_{A'A}(o)$, and by the symmetric property, $\bm \gamma_{AA'}(o) = \bm \gamma_{A'A}(\tau+1-o)$. Then,
the second part $\delta(\bm \gamma>\epsilon)$ for each term can be expressed as:

\figbox{Simulator $\SIM_A$}{

{\bf Evaluation:} $\PA$ does not receive message in this stage, and thus $\SIM_A$ does nothing.

{\bf Update: } 
\begin{itemize}
  \item $\SIM_A$ samples $\tilde{b}_j \leftarrow \GG$ for $j = 1, 2, \cdots, n'$.
  \item $\PA \leftarrow \SIM_A : \{ \tilde{b}_j \}_{j=1}^{n'}$ as $\PA$'s received message in Step (5) in $\Prot{\PRF}^{\sf DH}$.
\end{itemize}
}{The simulator for $\PA$ in the $\Prot{\PRF}^{\sf DH}$ protocol \label{fig:sim_A}.}

\figbox{Simulator $\SIM_B$}{

{\bf Evaluation:} 
\begin{itemize}
   \item $\SIM_B$ samples $\tilde{a}_i \leftarrow \GG$ for $i = 1, 2, \cdots, n$.
   \item $\SIM_B \rightarrow \PB: \{\tilde{a}_i \}_{i=1}^{n}$ as $\PB$'s received messages in Step (2).
\end{itemize}

{\bf Update: } 
\begin{itemize}
   \item $\SIM_B \leftarrow \PB:$ $\{b_j \in \GG\}_{j=1}^{n'}$ where $b_j$s are chosen by $\PB$.
   \item $\SIM_B \rightarrow \PB:$ $\{b_j^r \}_{j=1}^{n'}$ using random $r \leftarrow \ZZ_q$ as $\PB$'s received messages in Step (7).
\end{itemize}
}{The simulator for $\PB$ in the $\Prot{\PRF}^{\sf DH}$ protocol. \label{fig:sim_B}}

\begin{equation*}
    \delta(\bm \gamma_{AA'}>\epsilon)=\sum_{o = 1}^{\bm \gamma_{AA'}^{-1}(\epsilon)}(1-e^{\epsilon-\bm \gamma_{AA'}(o)})\Pr(\mathcal{M}(A) = o)
\end{equation*}
and 
\begin{equation*}
\begin{aligned}
&\delta(\bm \gamma_{A'A}>\epsilon)\\
=&\sum_{o = \bm \gamma_{A'A}^{-1}(\epsilon)}^{\tau}(1-e^{\epsilon-\bm \gamma_{A'A}(o)})\Pr(\mathcal{M}(A^{\prime}) = o)\\
=&\sum_{o = 1}^{\bm \gamma_{AA'}^{-1}(\epsilon)}(1-e^{\epsilon-\bm \gamma_{AA'}(o)})\Pr(\mathcal{M}(A) = o).
\end{aligned}
\end{equation*}
Therefore, $\delta(\bm \gamma_{AA'}>\epsilon)=\delta(\bm \gamma_{A'A}>\epsilon)$. 
\end{proof}

\subsubsection{Proof of Property~\ref{prop: window}}
\label{app:window_len}
\begin{proof}
$W$ is lower bounded by the maximum leakage after $k$-fold convolution. From \eqref{eq:PLD1}, the maximum leakage of the mechanism can be expressed as
\begin{equation*}
2\ln\left(\max_{z\in[0,\tau-1]} \frac{z+1}{\tau-z}\right)=2\ln(\tau)    
\end{equation*}
The length of the PLD is:
\begin{equation*}
2\ln\left(\max_{z\in[0,\tau-1]} \frac{z+1}{\tau-z}\right)-2\ln\left(\min_{z\in[0,\tau-1]} \frac{z+1}{\tau-z}\right)=4\ln{\tau}
\end{equation*}
Then, the maximum leakage after $k$-fold composition becomes:
\begin{equation*}
2\ln(\tau)+4(k-1)\ln(\tau)=(4k-2)\ln(\tau)
\end{equation*}
This completes the proof.
\end{proof}

\begin{figure*}
\centering
\includegraphics[width=\textwidth]{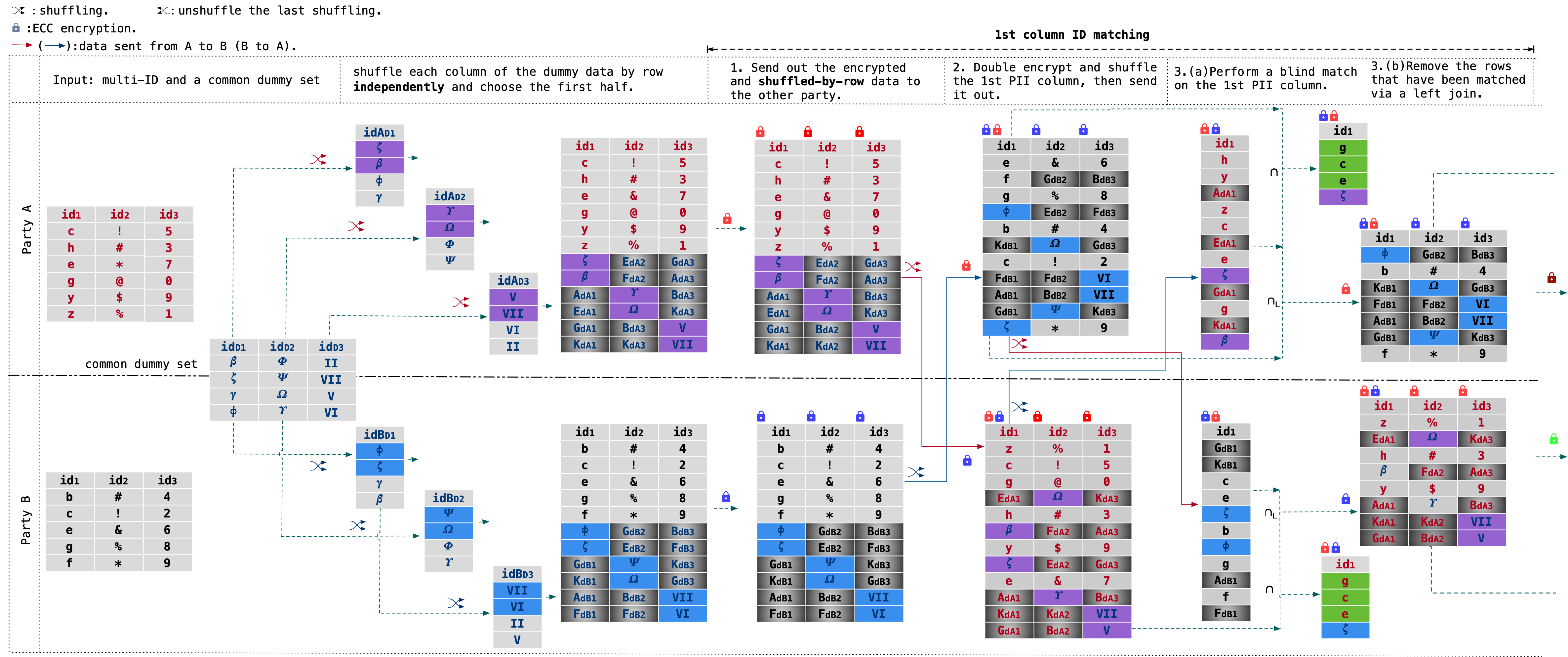}
\includegraphics[width=\textwidth]{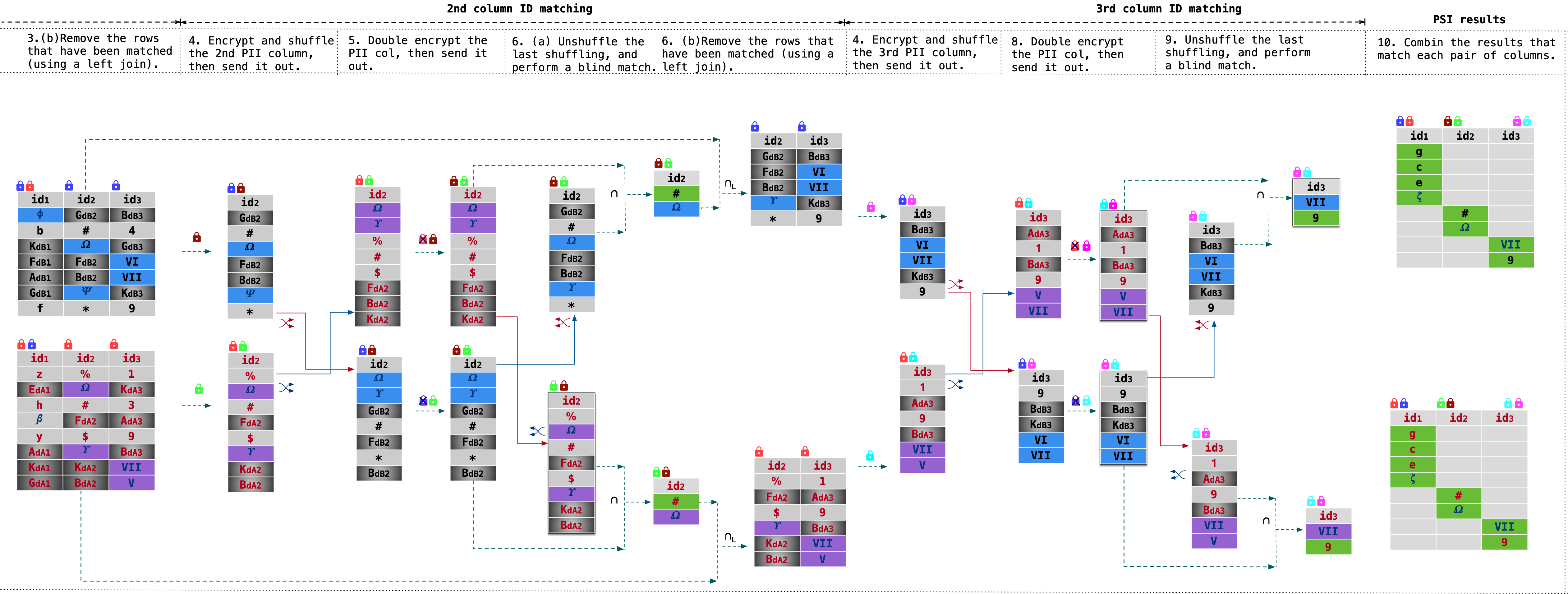}
\caption{An illustration of the DDH-based Waterfall Matching.}
\end{figure*}

\begin{figure*}
\centering
\fbox{
  \procedure[skipfirstln,linenumbering,codesize=\footnotesize]
  {Private Waterfall Matching and Sum Protocol $\Prot{wfm}$} 
  {
  \\
  \PA{}: \text{identifiers }\ID^1_A, \ID^2_A, \cdots, \ID^m_A \<\< \PB{}: \text{identifiers } \ID^1_B, \ID^2_B, \ID^m_B, \text{ and payloads } T
  \pclb \pcintertext[dotted]{Stage-1: Setups}
  \keyA{1}, \keyA{b}, \hat{k}^b_A \usample \mathbb{Z}_q \mbox{ for } b = 2, 3, \cdots, m\< \< 
  \keyB{1}, \keyB{b}, \hat{k}^b_B \usample \mathbb{Z}_q \mbox{ for } b = 2, 3, \cdots, m \\
  \F{\PRF}.\text{\bf init}(\sid_b, \keyA{b})\text{ for } b = 1, \cdots, m \< \< \F{\PRF}.\text{\bf init}(\sid_b, \keyB{b}) \text{ for } b = 1, \cdots, m\\
  \F{\PRF}.\text{\bf init}(\hat{\sid}_b, \hat{k}^b_A)\text{ for } b = 2, \cdots, m \< \< \F{\PRF}.\text{\bf init}(\hat{\sid}_b, \hat{k}^b_B) \text{ for } b = 2, \cdots, m\\
  \<  \< (\sk, \pk) \leftarrow \AGen(1^\lambda) \\
  \< \sendmessageleft*{\pk, C} \< C = \{ \AEnc(T[j]) \}_{j=1}^{|T|}
  \pclb
  \pcintertext[dotted]{Stage-2: PRF on all identifiers}
  \F{\PRF}.\text{\bf Send}(\sid_b, {\ID}^b_A) \text{ for } b = 1, \cdots, m \< \< {\sf Tag}^b_A \leftarrow \F{\PRF}.\text{\bf Recv}(\sid_b)\text{ for }b = 1, \cdots, m\\
  {\sf Tag}^b_B \leftarrow \F{\PRF}.\text{\bf Recv}(\sid_b)\text{ for }b = 1, \cdots, m \< \< \F{\PRF}.\text{\bf Send}(\sid_b, {\ID}^b_B)\text{ for }b = 1, \cdots, m
  \pclb \pcintertext[dotted]{Stage-3: Matching on the first identifier}
  \< \sendmessageleft*{\widetilde{{\sf Tag}}^1_A \text{ in shuffled order}} \<  \\
  \< \sendmessageright*{ \widetilde{{\sf Tag}}^1_{B} \text{ in shuffled order}} \< \\
  \mbox{Matches: } {J}^1_{B} \leftarrow \{ j: {\sf Tag}^1_{B}[j] \in \widetilde{{\sf Tag}}^1_{A} \}
  \<\< {J}^1_{A} \leftarrow \{j: {\sf Tag}^1_{A}[j] \in \widetilde{\sf Tag}^1_{B} \}\\
  \mbox{Mismatches: }{D}^1_B \leftarrow \{1, 2, \cdots, |{\sf Tag}^1_B|\} \setminus {J}^1_B \<\<  {D}^1_A \leftarrow \{1, 2, \cdots, |{\sf Tag}^1_A|\} \setminus {J}^1_A \\
  \mbox{Book-keep: } \sigma^1_B[j] \leftarrow {D}^1_B[j] \mbox{ for } 1\le j \le |{D}^1_B| \< \< 
  \sigma^1_A[j] \leftarrow {D}^1_A[j] \mbox{ for } 1\le j \le |{D}^1_A| 
  \pclb
  \pcintertext[dotted]{Stage-4: Key-update then matching on the $b$-th identifier ($b = 2, \cdots, m)$} 
  \text{Keep the unmatched records } {\sf rTag}^b_B \leftarrow \{{\sf Tag}^b_B[i] \}_{i \in {D}^{b-1}_B} \< \< {\sf rTag}^b_A \leftarrow \{{\sf Tag}^b_A[i] \}_{i \in {D}^{b-1}_A} \\
  \text{(Skip this if $b = m$) } {\sf rTag}^b_B \leftarrow \F{\PRF}.{\bf Update}(\sid_b, \hat{\sid}_b, {\sf rTag}^b_B) \< \< 
  \text{(Skip this if $b = m$) } \F{\PRF}.{\bf Assit}(\sid_b, \hat{\sid}_b)\\
  \F{\PRF}.{\bf Assit}(\sid_b, \hat{\sid}_b) \< \< {\sf rTag}^b_A \leftarrow \F{\PRF}.{\bf Update}(\sid_b, \hat{\sid}_b, {\sf rTag}^b_A) \\
  \< \sendmessageleft*{\widetilde{{\sf rTag}}^b_{A} \text{ in shuffled order}} \\
  \text{(Skip this line if $b = m$) } \< \sendmessageright*{\widetilde{{\sf rTag}}^b_{B} \text{ in shuffled order}} \\
  \text{Matches: } {J}^b_{B} \leftarrow \{ \sigma^{b-1}_B[j]: {\sf rTag}^b_{B}[j] \in \widetilde{{\sf rTag}}^b_{A} \} \< \< {J}^b_{A} \leftarrow \{ \sigma^{b-1}_A[j]: {\sf rTag}^b_{A}[j] \in \widetilde{{\sf rTag}}^b_{B} \} \\
  \mbox{Mismatches: }{D}^b_B \leftarrow D^{b-1}_B \setminus{J^b_B} \<\<  {D}^b_A \leftarrow D^{b-1}_A \setminus {J}^b_A \\
  \mbox{Book-keep: } \sigma^b_B[j] \leftarrow {D}^b_B[j] \mbox{ for } 1\le j \le |{D}^b_B| \< \< 
  \sigma^b_A[j] \leftarrow {D}^b_A[j] \mbox{ for } 1\le j \le |{D}^b_A| 
\pclb \pcintertext[dotted]{Final stage: aggregation}
  \mbox{Sum }{\sf ct} \leftarrow {\ASum}(\{ \widetilde{C}[j] \}_{j \in {J}^1_B \cup \cdots \cup {J}^b_B}) \<\sendmessageright*{{\sf ct}' \leftarrow {\AFresh}({\sf ct})}\< \text{Output } \ADec({\sf ct}'). }
}
\caption{(Full version) Our private waterfall matching protocol under $\F{\PRF}$-hybrid model.  $\mathcal{HE}$ is an AHE scheme.\label{fig:proposed_full_wfm}}
\end{figure*}

\subsubsection{Security Proof for $\Prot{\PRF}^{\sf DDH}$}
This section contains the full version of the security proof of $\Prot{\PRF}^{\sf DDH}$. 

Let $\VIEW_{\sigma, \lambda}^{\Prot{\PRF}^{\sf DDH}}(\{x_i\}_{i=1}^{n}, \emptyset)$
be a random variable representing the view of $P_\sigma$ ($\sigma = A, B$) in real protocol execution, 
where the random variable ranges over the internal randomness of all parties, and the randomness in the setup phase (including that of the Random Oracle).
The the view of a party consists of its internal state (including its input and randomness) and all messages this party received from the other party.
The messages sent by this party do not need to be part of the view because they can be determined using the other elements of its view.
\begin{align*}
&\VIEW_{A, \lambda}^{\Prot{\PRF}^{\sf DH}}(\{x_i\}_{i=1}^{n}) = (\emptyset, \{ \hash(x_j)^{k'_B} \}_{j \in Z}) \\
&\VIEW_{B, \lambda}^{\Prot{\PRF}^{\sf DH}}(\{x_i\}_{i=1}^{n}) = (\{\hash(x_i)^{k_A}\}_{i=1}^{n}, \{ \hash(x_j)^{k'_A} \}_{j \in Z}) 
\end{align*}

We now show that \PA{}'s view can be simulated given only \PA{}'s input but not \PB{}'s input. The simulator $\SIM_A$ is given in Figure~\ref{fig:sim_A}.  
We argue that $\{\hash(x_j)^{k'_B} \}_{j \in Z}$ is indistinguishable from random group elements in $\GG$ assuming the hardness of DDH under the random oracle model.
Specifically, the existence of an efficient distinguisher $\mathcal{D}$ that outputs 0 when presented with $\tilde{b}_1, \cdots, \tilde{b}_{n'}$ and outputs 1 
when it observes $\{\hash(x_j)^{k'_B}\}_{j \in Z}$ allows us to construct a simulator $\SIM$ that violates the DDH assumption, as follows.

Upon receiving a DDH challenge $(g, g^x, g^y, g^z)$, $\SIM$ does the following: 
\begin{itemize}
    \item Selects random values $d_1, \cdots, d_{n'-2}$ from $\ZZ_q$.
    \item Answers queries for $\hash$ as follows: $\hash(x_{j_1}) = g, \hash(x_{j_2}) = g^x$, and $\hash(x_{j_k}) = g^{d_{k-2}}$ for $k > 2$ for consistency.
    \item Sends $\{g^y, g^z, (g^y)d_1, \cdots, (g^y)^{d_{n'-2}} \}$ to $\mathcal{D}$.
\end{itemize}
If $(g, g^x, g^y, g^z)$ is a Diffie-Hellman tuple, i.e., $z = xy$, then 
\begin{align*}
  &\{g^y, g^z, (g^y)^{d_1}, \cdots, (g^y)^{d_{n'-2}}\} \\
= &\{g^y, (g^x)^y, (g^{d_1})^y, \cdots, (g^{d_{n'-2}})^y\},
\end{align*}
which is distributed like $\{\hash(x_{j_1})^{k'_B}, \cdots, \hash(x_{j_{n'}})^{k'_B}\}$. 
Thus, $\mathcal{D}$ must return 1.
On the other hand, when $z \neq xy$, the term $g^z \neq \hash(x_{j_2})^y$. As a result, $\SIM$ can use $\mathcal{D}$'s output to respond to the DDH challenge correctly.

The simulator $\SIM_B$ is given in Figure~\ref{fig:sim_B}.  
For the evaluation part, the values in $\{\hash(x_i)^{k_A}\}$ are indistinguishable from random group elements under the DDH assumption.
Similarly, for the update part, the values in $\{\hash(x_j)^{k'_A}\}$ are also indistinguishable from random group elements in $\GG$.

\section{Full Protocol with More Identifiers}
Figure~\ref{fig:proposed_full_wfm} depicts the full version of our waterfall matching protocol that handles $m \ge 2$ identifiers.

\noindent{\bf Complexities.} 
Let the databases sizes being ${n}_A = |{\ID}^b_A|$ and ${n}_B = |\ID^b_B|$, and $|{J}^{b}|$ being the intersection size on the $b$-th identifier.
When instantiating $\F{\PRF}$ as the DDH-based construction, 
the computation of Figure~\ref{fig:proposed_full_wfm} includes $O(2m({n}_A + {n}_A) - 2\sum_{b=1}^{m-1} |{J}^{b}|)$ group operations, and $O({n}_B + \sum_{b=1}^{m}|J^b|)$ AHE operations.
The communication includes $O(m({n}_A + {n}_B) + {n}_B - \sum_{b=1}^{m-1}|{J}^{b}|)$ group elements, and $O({n}_B)$ AHE ciphertexts.

\end{document}